\documentclass[sigconf]{acmart}
\copyrightyear{2018}
\acmYear{2018}
\setcopyright{iw3c2w3}
\acmConference[WWW 2018]{The 2018 Web Conference}{April 23--27, 2018}{Lyon, France}
\acmPrice{}
\acmDOI{10.1145/3178876.3186003}
\acmISBN{978-1-4503-5639-8/18/04}
\usepackage[T1]{fontenc}
\usepackage{capt-of}
\usepackage{subcaption}
\usepackage{balance}
\usepackage{thm-restate}
\usepackage{enumitem}
\usepackage{xcolor}
\usepackage{multirow}
\usepackage{alltt}
\usepackage{algorithm}
\usepackage{array}
\usepackage[noend]{algpseudocode}
\usepackage{tikz}
\usetikzlibrary{fit,positioning,shapes}
\usepackage{pgfplots}
\usepgfplotslibrary{statistics}
\pgfplotsset{compat=1.5}
\usepackage{macros}
\usepackage{balance}
\usepackage{nopageno}
\allowdisplaybreaks

\withappendix

\setlist{topsep=2pt,parsep=1pt,itemsep=1pt}

\begin{document}

\title{Estimating the Cardinality of Conjunctive Queries \\ over RDF Data Using Graph Summarisation}
\author{Giorgio Stefanoni}
\affiliation{Bloomberg L.P., London, UK}
\author{Boris Motik}
\affiliation{University of Oxford, UK}
\author{Egor V.\ Kostylev}
\affiliation{University of Oxford, UK}

\begin{abstract}
Estimating the cardinality (i.e., the number of answers) of conjunctive queries
is particularly difficult in RDF systems: queries over RDF data are
navigational and thus tend to involve many joins. We present a new, principled
cardinality estimation technique based on graph summarisation. We interpret a
summary of an RDF graph using a possible world semantics and formalise the
estimation problem as computing the expected cardinality over all RDF graphs
represented by the summary, and we present a closed-form formula for computing
the expectation of arbitrary queries. We also discuss approaches to RDF graph
summarisation. Finally, we show empirically that our cardinality technique is
more accurate and more consistent, often by orders of magnitude, than the state
of the art.

\end{abstract}

\maketitle

\section{Introduction}\label{sec:introduction}

The Resource Description Framework (RDF) \cite{rdf11-concepts} data model is
often used in applications where developing a rigid schema is either infeasible
or undesirable. An RDF graph is a set of triples of the form $\tuple{s, p, o}$,
where $s$, $p,$ and $o$ are objects called \emph{resources}; by viewing triples
as directed labelled edges, an RDF graph corresponds to an ordinary directed
labelled graph. SPARQL \cite{sparql11-query} is the standard RDF query
language. \emph{Conjunctive queries} (also known as \emph{basic graph
patterns}) are the basic type of SPARQL queries, and estimating their
cardinality (i.e., the number of answers) is an important problem. For example,
estimates are used to determine the cost of query plans, so estimation accuracy
directly influences plan quality \cite{Leis:2015:GQO:2850583.2850594}. Thus,
cardinality estimators are key components of most RDF systems.

Cardinality estimation has a long tradition in databases. It is usually solved
by summarising a database to different kinds of \emph{synopses} that can be
used to accurately estimate the cardinality of certain queries. One-dimensional
synopses, such as one-dimensional histograms
\cite{Poosala:1996:IHS:233269.233342,Ioannidis:2003} and wavelets
\cite{Matias:1998:WHS:276304.276344}, can summarise one attribute of one
relation, so they can be used on queries involving one selection on a single
relation. Multidimensional synopses, such as multidimensional histograms
\cite{Poosala:1997, Bruno:2001, Gunopulos2005, Aboulnaga:1999:SHB},
multidimensional wavelets \cite{Chakrabarti2001, Garofalakis:2002}, discrete
cosine transforms \cite{Lee:1999:MSE}, and kernel methods
\cite{Gunopulos:2000:AMA}, can summarise several attributes of one relation, so
they can be used on queries involving several selections on a single relation.
Schema-level synopses, such as graphical models
\cite{Getoor:2001,Tzoumas:2013}, sampling
\cite{Acharya:1999,Yu:2013:CND:2463676.2463701}, TuG synopses
\cite{Spiegel:2009}, and statistical views \cite{DBLP:conf/sigmod/BrunoC04,
DBLP:conf/vldb/Galindo-LegariaJWW03, DBLP:conf/sigmod/TerleckiBGZ09} can
summarise several attributes over different relations, but they often require a
\emph{join schema} that identifies the supported joins.

Queries whose cardinality cannot be estimated using the available synopses are
typically broken into subqueries that can be estimated, and partial estimates
are combined using ad hoc assumptions
\cite{Bruno:2003,DBLP:books/ph/Garcia-MolinaUW99}: the \emph{independence
assumption} assumes that each selection or join affects the query answers
independently, the \emph{preservation assumption} assumes that each attribute
value of a relation participating in a join is present in the join result, and
the \emph{containment assumption} assumes that, for each pair of joined
attributes, all values of one attribute are contained in the other attribute.
However, these assumptions usually do not hold in practice and thus introduce
estimation errors that compound exponentially with the number of
joins~\cite{Ioannidis:1991:PES:115790.115835}. Due to the graph-like nature of
RDF, queries in RDF often navigate over paths and thus contain many
(self-)joins (ten or more are not rare). Techniques that aim to address these
specifics of RDF either just adapt relational approaches \cite{Stocker:2008,
Neumann2010, Huang:2011} or focus on particular types of queries (e.g., star or
chain queries) \cite{Neumann:2011} and fall back to the ad hoc assumptions for
other types of queries. Hence, as we show in Section~\ref{sec:evaluation},
accurately estimating the cardinality of complex RDF queries remains
challenging.

To address these issues, in Section~\ref{sec:framework} we propose a new,
principled technique for estimating the cardinality of conjunctive queries over
RDF data. Our technique is based on graph summarisation---the process of
compressing a graph by merging its vertices. Summaries have already been used
for graph exploration~\cite{Navlakha:2008:GSB:1376616.1376661,
Tian:2008:EAG:1376616.1376675}, fast approximate graph
analytics~\cite{DBLP:conf/sdm/LeFevreT10, DBLP:journals/datamine/RiondatoGB17},
and query processing~\cite{Cebiric:2015:QSR:2824032.2824124}. In
Section~\ref{sec:framework}, we introduce a specific kind of summary that we
use as a schema-level synopsis for RDF data. Following
\citet{DBLP:conf/sdm/LeFevreT10}, we interpret a summary using a \emph{possible
world semantics} as a family of RDF graphs represented by the summary. We
formalise the cardinality estimation problem as computing the expectation of
the query cardinality across this family, and so our approach does not require
any ad hoc assumptions. Finally, we show how to determine the probability that
the estimation error exceeds a given bound. To the best of our knowledge, this
is the first approach that can provide such guarantees for arbitrary queries.

In Section~\ref{sec:estimation} we present formulas for computing the expected
cardinality and its variance. Our formulas can handle queries of the form
\texttt{SELECT * WHERE \{ BGP \}} for \texttt{BGP} a basic graph pattern, which
form the core of SPARQL. Note that, without \texttt{DISTINCT}, projecting
variables in the \texttt{SELECT} clause does not affect the query cardinality.

To apply our framework in practice, effective graph summarisation algorithms
are needed. Any summarisation algorithm can be used with our technique in
principle, but some will fit our framework better than others: we intuitively
expect that a good summary should merge vertices participating in similar
connections. Hence, algorithms for unlabelled graphs
\cite{Navlakha:2008:GSB:1376616.1376661, DBLP:conf/sdm/LeFevreT10,
DBLP:journals/datamine/RiondatoGB17} are unlikely to work well on RDF, where
links are inherently labelled. We experimented with an adaptation of the SNAP
\cite{Tian:2008:EAG:1376616.1376675} technique to RDF, but it produced very
large summaries. A recently proposed RDF-specific summarisation technique seems
very promising, but its computational complexity seems very high
\cite{Cebiric:2015:QSR:2824032.2824124}. Thus, in
Section~\ref{sec:summarisation} we present a new graph summarisation approach.
In the first step, we use a simple notion of similarity based on the labels on
the edges that resources occur in. If the resulting summary is not sufficiently
small, in the second step we further merge resources based on the similarity of
the connections resources participate in.

In Section~\ref{sec:evaluation} we present the results of an extensive
evaluation on LUBM, WatDiv, and DBLP benchmarks with 55~M to 109~M triples, and
15 to 103 queries. We compared the accuracy of our technique against RDF-3X
\cite{Neumann2010} and the \emph{characteristic sets technique}
\cite{Neumann:2011} (both specifically targeting RDF data), PostgreSQL (which
uses one-dimensional histograms), and SystemX (a prominent commercial RDBMS
that combines histograms with dynamic sampling); we considered both triple
table and vertical partitioning for storing RDF in RDBMSs. We show that our
technique is generally more accurate than the related approaches, often by
orders of magnitude.

The proofs of all of our technical results are given \ifappendix{in full in the
appendices.}{in the extended version of the paper at
\url{http://arxiv.org/abs/1801.09619}.}

\section{Definitions and Notation}\label{sec:preliminaries}

We now recall the relevant definitions and notation. RDF vocabulary consist of
\emph{resources}, which are URIs (i.e., abstract entities) or \emph{literals}
(i.e., concrete values described by datatypes such as
$\mathit{xsd}{:}\mathit{string}$ or $\mathit{xsd}{:}\mathit{integer}$). A
\emph{term} is a resource or a variable. An \emph{(RDF) atom} is a triple of
the form ${\tuple{s, p, o}}$, where $s$, $p$, and $o$ are terms called the
\emph{subject}, \emph{predicate}, and \emph{object}, respectively. An
\emph{(RDF) triple} is a variable-free atom. An \emph{(RDF) graph} is a finite
set of triples.

A \emph{substitution} $\pi$ is a mapping of variables to terms; the domain and
range of $\pi$ are $\dom{\pi}$ and $\rng{\pi}$, respectively; $\vrng{\pi}$ is
the set of all variables in $\rng{\pi}$; and $\pi$ is \emph{empty} if
${\dom{\pi} = \emptyset}$. Let $Z$ be a term, an atom, or a set of atoms. Then,
$\pi(Z)$ is obtained from $Z$ by replacing each occurrence of ${x \in
\dom{\pi}}$ in $Z$ with ${\pi(x)}$; if $Z$ is a set of atoms, then $\pi(Z)$ is
also a set and does not contain duplicates. Moreover, $\res{Z}$, $\var{Z}$, and
$\term{Z}$ are the sets of resources, variables, and terms occurring in $Z$,
respectively.

SPARQL \cite{sparql11-query} is a standard query language for RDF. Its syntax
is very verbose, so we use a more compact notation that captures basic SPARQL
queries of the form ${\texttt{SELECT * WHERE \{ BGP \}}}$. SPARQL variables are
written as $?x$, but we drop the question mark and reserve (possibly indexed)
letters $x$, $y$, and $z$ for variables. A \emph{conjunctive query} (\emph{CQ})
$q$ is a finite set of atoms. A substitution $\pi$ is an \emph{answer} to $q$
on an RDF graph $\graph$ if ${\dom{\pi} = \var{q}}$ and ${\pi(q) \subseteq
\graph}$; moreover, $\ans{q}{\graph}$ is the set of all answers to $q$ on
$\graph$, and the \emph{cardinality} of $q$ on $\graph$ is the size of
$\ans{q}{\graph}$. For technical reasons, we allow $q$ to be empty, in which
case $\ans{q}{G}$ contains just one empty answer. These definitions are
compatible with the ones by \citet{DBLP:journals/tods/PerezAG09}; however,
since we consider only CQs without variable projection, we simplify the
notation and treat the query as a finite set of atoms.

The q-error is often used to measure precision of cardinality estimates
\cite{Neumann:2011}: if a query returns $N$ answers and its cardinality is
estimated as $E$, the \emph{q-error} of the estimate is ${\qerror =
\max(\myfrac{N'}{E'},\myfrac{E'}{N'})}$, where ${N' = \max(N,1)}$ and ${E' =
\max(E,1)}$. Intuitively, the q-error shows the difference in the orders of
magnitude between a real cardinality $N$ and its estimation $E$, regardless of
whether $E$ over- or undershoots $N$. We use $N'$ and $E'$ instead of $N$ and
$E$ in the definition of q-error in order to avoid division by zero, as well as
to prevent artificially high q-error values when ${E < 1}$.

Let ${n \geq m > 0}$ be integers. Then, $\binom{n}{m} = \frac{n!}{m! (n-m)!}$
is the \emph{binomial coefficient}, ${(n)_m = n(n-1) \dots (n - m +1)}$ is the
\emph{falling factorial}, and ${\interval{m}{n}}$ is the set of all integers
between $m$ and $n$ (inclusive).

\section{Cardinality Estimation Framework}\label{sec:framework}

\begin{figure*}[tb]
    \centering
    \newenvironment{customlegend}[1]{
    \begingroup
    \csname pgfplots@init@cleared@structures\endcsname
    \pgfplotsset{#1}
}{
    \csname pgfplots@createlegend\endcsname
    \endgroup
}
\def\addlegendimage{\csname pgfplots@addlegendimage\endcsname}
\pgfkeys{/pgfplots/number in legend/.style={
        /pgfplots/legend image code/.code={
            \node at (0.125,-0.0225){#1};
        },
    },
}
\tikzset{
    align at top/.style={baseline=(current bounding box.north)},
    x = 1cm,
    y = 1.1cm,
    font = \tiny,
    graph vertex/.style={
        draw,
        minimum size = 11pt,
        inner sep = 0.7pt,
    },
    graph married/.style = {
        graph vertex,
        circle,
        fill = gray!20,
    },
    graph single/.style = {
        graph vertex,
        diamond,
        fill = red!20,
    },
    graph road/.style = {
        graph vertex,
        rectangle,
        fill = blue!20,
    },
    graph van/.style = {
        graph vertex,
        regular polygon,
        regular polygon sides=5,
        fill = green!20,
    },
    graph edge owns/.style= {
        ->,
    },
    graph edge boss/.style= {
        ->,
        dashed,
    }
}
\begin{tikzpicture}[]
    \node[graph single]     (s1)       at (0, 0)    {\footnotesize{$e_1$}};
    \node[graph single]     (s2)       at (1, 0)    {\footnotesize{$e_2$}};
    \path[graph edge boss]  (s1) edge node[above]   {}  (s2) ;
    \node[graph married]    (m1)       at (0, -0.8) {\footnotesize{$e_3$}};
    \node[graph married]    (m2)       at (1, -0.8) {\footnotesize{$e_4$}};
    \path[graph edge boss]  (s1) edge node[above]   {}  (m1) ;
    \path[graph edge boss]  (s2) edge node[above]   {}  (m2) ;
    \node[graph van]        (v1)       at (0, -1.6) {\footnotesize{$c_3$}};
    \node[graph van]        (v2)       at (1, -1.6) {\footnotesize{$c_4$}};
    \path[graph edge owns]  (m1) edge node[above]   {}  (v1) ;
    \path[graph edge owns]  (m2) edge node[above]   {}  (v2) ;
    \node[graph road]       (r1)       at (2, 0)    {\footnotesize{$c_1$}};
    \node[graph road]       (r2)       at (2, -0.8) {\footnotesize{$c_2$}};
    \path[graph edge owns]  (s2) edge node[above]   {}  (r1) ;
    \path[graph edge owns]  (m2) edge node[above]   {}  (r2) ;
    \node[rounded corners=0.1cm, draw, dotted, fit=(s1) (s2)] (S) {};
    \node[left] at (S.west) {\footnotesize{$b_1$}};
    \node[rounded corners=0.1cm, draw, dotted, fit=(m1) (m2)] (M) {};
    \node[left] at (M.west) {\footnotesize{$b_3$}};
    \node[rounded corners=0.1cm, draw, dotted, fit=(v1) (v2)] (V) {};
    \node[left] at (V.west) {\footnotesize{$b_4$}};
    \node[rounded corners=0.1cm, draw, dotted, fit=(r1) (r2)] (R) {};
    \node[right] at (R.east) {\footnotesize{$b_2$}};

    \node[] at (1, -2.1) {\footnotesize{(a)}};

    \node[graph single]     (B1)       at (4, 0)    {\footnotesize{$b_1$}};
    \node[graph married]    (B3)       at (4, -0.8) {\footnotesize{$b_3$}};
    \node[graph van]        (B4)       at (4, -1.6) {\footnotesize{$b_4$}};
    \node[graph road]       (B2)       at (5, 0)    {\footnotesize{$b_2$}};
    \path[graph edge boss]  (B1) edge[loop left] node[left] {\footnotesize{1}}  (B1) ;
    \path[graph edge owns]  (B1) edge node[below]   {\footnotesize{1}}  (B2) ;
    \path[graph edge owns]  (B3) edge node[right]   {\footnotesize{1}}  (B2) ;
    \path[graph edge boss]  (B1) edge node[left]    {\footnotesize{2}}  (B3) ;
    \path[graph edge owns]  (B3) edge node[left]    {\footnotesize{2}}  (B4) ;

    \node[] at (4.5, -2.1) {\footnotesize{(b)}};

    \node[graph single]     (s1a)      at (7, 0)    {\footnotesize{$e_1$}};
    \node[graph single]     (s2a)      at (8, 0)    {\footnotesize{$e_2$}};
    \path[graph edge boss]  (s2a) edge node[above]  {}  (s1a) ;
    \node[graph married]    (m1a)      at (7, -0.8) {\footnotesize{$e_3$}};
    \node[graph married]    (m2a)      at (8, -0.8) {\footnotesize{$e_4$}};
    \path[graph edge boss]  (s1a) edge node[above]  {}  (m1a) ;
    \path[graph edge boss]  (s2a) edge node[above]  {}  (m1a) ;
    \node[graph van]        (v1a)      at (7, -1.6) {\footnotesize{$c_3$}};
    \node[graph van]        (v2a)      at (8, -1.6) {\footnotesize{$c_4$}};
    \path[graph edge owns]  (m2a) edge node[above]  {}  (v1a) ;
    \path[graph edge owns]  (m2a) edge node[above]  {}  (v2a) ;
    \node[graph road]       (r1a)      at (9, 0)    {\footnotesize{$c_1$}};
    \node[graph road]       (r2a)      at (9, -0.8) {\footnotesize{$c_2$}};
    \path[graph edge owns]  (s2a) edge node[above]  {}  (r1a) ;
    \path[graph edge owns]  (m2a) edge node[above]  {}  (r2a) ;

    \node[graph single]     (s1b)     at (10, 0)    {\footnotesize{$e_1$}};
    \node[graph single]     (s2b)     at (11, 0)    {\footnotesize{$e_2$}};
    \path[graph edge boss]  (s2b) edge node[above]  {}  (s1b);
    \node[graph married]    (m1b)     at (10, -0.8) {\footnotesize{$e_3$}};
    \node[graph married]    (m2b)     at (11, -0.8) {\footnotesize{$e_4$}};
    \path[graph edge boss]  (s1b) edge node[above]  {}  (m1b) ;
    \path[graph edge boss]  (s2b) edge node[above]  {}  (m1b) ;
    \node[graph van]        (v1b)     at (10, -1.6) {\footnotesize{$c_3$}};
    \node[graph van]        (v2b)     at (11, -1.6) {\footnotesize{$c_4$}};
    \path[graph edge owns]  (m1b) edge node[above]  {}  (v1b) ;
    \path[graph edge owns]  (m1b) edge node[above]  {}  (v2b) ;
    \node[graph road]       (r1b)     at (12, 0) {\footnotesize{$c_1$}};
    \node[graph road]       (r2b)     at (12, -0.8) {\footnotesize{$c_2$}};
    \path[graph edge owns]  (s2b) edge node[above]  {}  (r1b) ;
    \path[graph edge owns]  (m2b) edge node[above]  {}  (r2b) ;

    \node[] at (9.5, -2.1) {\footnotesize{(c)}};

    \begin{customlegend}{
        legend cell align=left,
        legend entries={
            $\mathit{owns}$,
            $\mathit{manages}$,
            emp.\ of class Single,
            emp.\ of class Married,
            car of class Roadster,
            car of class Van,
        },
        legend style={
            at={(16,0.2)},
            font=\footnotesize,
            draw=none,
            mark size=3pt,
            mark options={draw=black}
        }
    }
        \addlegendimage{->,black}
        \addlegendimage{->,black,dashed}
        \addlegendimage{mark =diamond*,  draw=white,  fill=red!20}
        \addlegendimage{mark =*, draw=white, fill=gray!20}
        \addlegendimage{mark =square*, draw=white, fill=blue!20}
        \addlegendimage{mark =pentagon*, draw=white, fill=green!20}
    \end{customlegend}

    \node[] at (14, -2.1) {\footnotesize{Legend}};

\end{tikzpicture}
    \caption{An example RDF graph $\graph$ (a), its summary $\summary$ (b), and example graphs $\graph_1$ and $\graph_2$ that also summarise as $\summary$ (c)}\label{fig:summary}
\end{figure*}
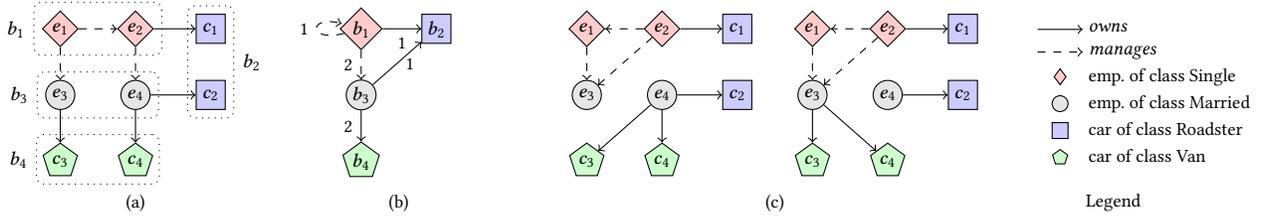

To estimate the cardinality of a query without ad hoc assumptions, we must
`compress' our input RDF graph while preserving as much of its structure as
possible. Graph summarisation has already been used to this end in related
settings~\cite{Navlakha:2008:GSB:1376616.1376661,
Tian:2008:EAG:1376616.1376675, DBLP:conf/sdm/LeFevreT10,
DBLP:journals/datamine/RiondatoGB17, Cebiric:2015:QSR:2824032.2824124}, and in
our work we use it to estimate CQ cardinality in a principled way.

Consider the input graph $\graph$ in Figure~\ref{fig:summary} where employees
are connected to their cars and managers, and resources are assigned to classes
(using the $\rdftype$ predicate) as shown in the legend. A summary $\summary$
of $\graph$ is obtained by merging the resources of $G$ into \emph{buckets} as
shown using dotted boxes; for example, employees $e_1$ and $e_2$ are replaced
by bucket $b_1$, and $e_3$ and $e_4$ by another bucket $b_3$. In contrast to
the existing approaches, in our approach each edge in $\summary$ is labelled
with the number of edges of $\graph$ that collapse due to merging. For example,
the $\mathit{manages}$ edge from $e_1$ to $e_2$ collapses into a self-loop on
$b_1$ of weight $1$, and the edges from $e_1$ to $e_3$ and from $e_2$ to $e_4$
collapse into an edge from $b_1$ to $b_3$ of weight $2$.

Our objective is to estimate the cardinality of a CQ $q$ on $\graph$ using
$\summary$, rather than $\graph$. To this end, we interpret $\summary$ using
the `possible worlds' semantics~\cite{DBLP:conf/sdm/LeFevreT10}: we assume that
$\summary$ represents with equal probability each graph $\graph'$ that
summarises as $\summary$---that is, if we replace the resources of $\graph'$
with buckets in the same way as when we constructed $\summary$ from $\graph$,
we obtain $\summary$. Thus, $\summary$ represents the input graph $\graph$,
but, due to the loss of information in summarisation, it represents other
graphs as well. In Figure~\ref{fig:summary}, in addition to $\graph$, summary
$\summary$ also represents graphs $\graph_1$ and $\graph_2$. We then estimate
the cardinality of $q$ on $\graph$ as the expected cardinality of $q$ over all
graphs represented by $\summary$. Instead of the assumptions mentioned in
Section~\ref{sec:introduction}, we thus use a consistent semantic
interpretation of $\summary$.

We next formalise our notion of a summary. We do not talk of a `summary of an
input graph $\graph$': a summary is a synopsis that represents a family of
graphs.

\begin{definition}\label{def:summary}\em
    A \emph{summary} ${\summary = \summarytuple}$ is a triple where $\sgraph$
    is an RDF graph called the \emph{summarisation graph}, ${\weightfunction :
    \sgraph \to \nat}$ is a \emph{weight function} assigning a positive
    natural number to each triple in $\sgraph$, and $\mu$ is a surjective
    \emph{summarisation function} from a finite set of resources $\dom{\mu}$ to
    $\res{\sgraph}$. The resources of ${\res{\sgraph}}$ are called
    \emph{buckets}.

    For $Z$ a term, an atom, or a set of atoms, $\mu(Z)$ is obtained from $Z$
    by replacing each resource ${\gres \in \res{Z} \cap \dom{\mu}}$ with
    ${\mu(\gres)}$; if $Z$ is a set, then $\mu(Z)$ is a set as well (i.e.,
    duplicates are removed). The $\summary$-\emph{size} of a bucket ${b \in
    \res{\sgraph}}$ is defined as ${\size{b} = |\setof{\gres \in \dom{\mu} \mid
    \mu(\gres) = b}|}$, and the $\summary$-\emph{size} of a triple
    ${\tuple{b_1,b_2,b_3} \in \res{\sgraph} \times \res{\sgraph} \times
    \res{\sgraph}}$ is defined as ${\size{\tuple{b_1,b_2,b_3}} = \size{b_1}
    \times \size{b_2} \times \size{b_3}}$.
\end{definition}

While buckets in histograms record value ranges, in our approach $\mu$
explicitly associates buckets with resources; this is needed as there is no
natural notion of a range of RDF resources. In Figure~\ref{fig:summary}, we
have ${\mu(e_1) = \mu(e_2) = b_1}$, ${\mu(c_1) = \mu(c_2) = b_2}$, ${\mu(e_3) =
\mu(e_4) = b_3}$, ${\mu(c_3) = \mu(c_4) = b_4}$, and $\mu$ is identity on all
remaining resources, all of which occur in $\graph$ as a predicate or the
object of a triple whose predicate is $\rdftype$ (such as $\mathit{owns}$,
$\rdftype$, or $\mathit{Van}$).

The size of a bucket $b$ is the number of resources mapped to $b$; in our
example, ${\size{b_i} = 2}$ for ${i \in \interval{1}{4}}$, but
${\size{\rdftype} = \size{\mathit{Van}} = 1}$. For ${\gtriple \in \graph}$ and
${\striple \in \sgraph}$ where ${\mu(\gtriple) = \striple}$, we say that
$\gtriple$ is \emph{summarised} as $\striple$, and that $\striple$ can be
\emph{expanded into} $\gtriple$. Clearly, $\striple$ can be expanded into at
most $\size{\striple}$ triples. In our example, we have
\begin{displaymath}
\begin{array}{@{}r@{\;}l@{\;}r@{\;}l@{}}
    \size{\tuple{b_1, \mathit{manages}, b_1}}   & = \size{b_1} \times \size{\mathit{manages}} \times \size{b_1}     & = 4, \\
    \size{\tuple{b_4, \rdftype, \mathit{Van}}}  & = \size{b_4} \times \size{\rdftype} \times \size{\mathit{Van}}    & = 2. \\
\end{array}
\end{displaymath}

Definition~\ref{def:semantics} formalises the `possible worlds' semantics of
$\summary$.

\begin{definition}\label{def:semantics}\em
    A summary ${\summary = \summarytuple}$ \emph{represents} a graph $\graph$
    if (i)~${\res{\graph} \subseteq \dom{\mu}}$, (ii)~${\sgraph =
    \mu(\graph)}$, and (iii)~for each ${\striple \in \sgraph}$, it is the case
    that ${\weight{\striple} = |\setof{\gtriple \in \graph \mid \mu(\gtriple) =
    \striple}|}$. Moreover, $\semantics$ is the set of all graphs that
    $\summary$ represents. Finally, summary $\summary$ is \emph{consistent} if
    ${\semantics \neq \emptyset}$.
\end{definition}

At most $\size{\striple}$ triples can be summarised as ${\striple \in
\sgraph}$, so ${\weight{\striple} \leq \size{\striple}}$ must hold for
$\summary$ to be consistent. Also, a graph in $\semantics$ is obtained by
expanding each ${\striple \in \sgraph}$ arbitrarily into $\weight{\striple}$
triples if ${\weight{\striple} \leq \size{\striple}}$ holds. Thus, we can check
the consistency of $\summary$ as in Proposition~\ref{prop:consistency-checking}.

\begin{proposition}\label{prop:consistency-checking}
    A summary ${\summary = \summarytuple}$ is consistent if and only if
    ${\weight{\striple} \leq \size{\striple}}$ holds for each ${\striple \in
    \sgraph}$.
\end{proposition}

We now estimate the cardinality of a CQ $q$ as the expected cardinality of $q$
over all graphs represented by $\summary$. We also introduce the related,
natural notion of cardinality variance.

\begin{definition}\label{def:expect-var}\em
    The \emph{expectation} $\expected{q}$ and the \emph{variance}
    $\variance{q}$ of the cardinality of a CQ $q$ on a summary $\summary$ are
    defined as
    \begin{displaymath} 
        \expected{q} = \sum_{\graph \in \semantics}\frac{|\ans{q}{\graph}|}{|\semantics|} \;\; \text{and} \;\; \variance{q} = \sum_{\graph \in \semantics}\frac{(|\ans{q}{\graph}| - \expected{q})^2}{|\semantics|}.
    \end{displaymath}
\end{definition} 

Please note that $\expected{q}$ is not necessarily a natural number. For
example, if $q$ is a variable-free query, then $|\ans{q}{\graph}|$ is equal to
the number of graphs in $\semantics$ that contain $q$; since
${|\ans{q}{\graph}| \leq |\semantics|}$ is always satisfied, $\expected{q}$ is
a rational number satisfying ${0 \leq \expected{q} \leq 1}$.

Using $\expected{q}$ and $\variance{q}$, we can bound from above the
probability $\prob{\qerror \geq \varepsilon}$ that the q-error of the estimate
exceeds a given amount $\varepsilon$ as shown in Theorem~\ref{th:qerrorbound}.
Intuitively, for a graph $\graph'$ chosen from $\semantics$ uniformly at
random, $\prob{\qerror \geq \varepsilon}$ is the probability that
$\expected{q}$ over- or underestimates $|\ans{q}{\graph'}|$ by a factor larger
than $\varepsilon$; we can equivalently express this by saying that
${\prob{\qerror \geq \varepsilon} \cdot |\semantics|}$ is the number of such
graphs $G'$. We show in Section~\ref{sec:evaluation:bounds} that $\prob{\qerror
\geq \varepsilon}$ is often small even for $\varepsilon$ as low as $10$, so in
such cases there is good chance that the cardinality of $q$ on our input graph
is close to $\expected{q}$.

\begin{restatable}{theorem}{qerrorbound}\label{th:qerrorbound}
    Let $q$ be a CQ and let $\summary$ be a summary such that ${\expected{q}
    \geq 1}$. Then, for each ${\varepsilon > 1}$ and for $\qerror$ the q-error
    of estimating the cardinality of $q$ as $\expected{q}$, we have
    \begin{displaymath}
        \prob{\qerror \geq \varepsilon} \leq \left(\frac{\varepsilon \cdot \std{q}}{(\varepsilon - 1) \cdot \expected{q}}\right)^2.
    \end{displaymath}
    If ${\varepsilon \geq \expected{q}}$ holds additionally, then the numerator
    ${\varepsilon \cdot \std{q}}$ on the right-hand side of the inequality can
    be replaced with $\std{q}$.
\end{restatable}

\section{Computing the Expectation}\label{sec:estimation}

Given a CQ $q$ and a summary ${\summary = \summarytuple}$, we can compute
$\expected{q}$ by iterating over all graphs in $\semantics$, but this is
impractical since the number of such graphs is exponential in $|\sgraph|$. In
this section, we present formulas that compute $\expected{q}$ in time
polynomial in $|\sgraph|$.

\subsection{Intuition}\label{sec:estimation:intuition}

Our formulas are quite complex, so we first discuss the intuitions using the
summary $\summary$ for Figure~\ref{fig:summary} and several example queries.

We first explain how to compute $\expected{q}$ for variable-free queries, which
is the basic element of our approach. Let ${q_1 = \setof{\gtriple_1,
\gtriple_2}}$ where ${\gtriple_1 = \tuple{e_1, \mathit{manages}, e_3}}$ and
${\gtriple_2 = \tuple{e_3, \mathit{owns}, c_3}}$. Since $q_1$ is variable-free,
for each ${\graph' \in \semantics}$, the cardinality ${\ans{q_1}{\graph'}}$ of
$q_1$ on $\graph'$ is one if ${q_1 \subseteq \graph'}$, and it is zero
otherwise; thus, the numerator of $\expected{q_1}$ is equal to the number of
graphs in $\semantics$ that contain all atoms of $q_1$ (i.e., that contain
$q_1$ as a subset). Triples $\gtriple_1$ and $\gtriple_2$ of the query are
summarised as ${\striple_1 = \tuple{b_1, \mathit{manages}, b_3}}$ and
${\striple_2 = \tuple{b_3, \mathit{owns}, b_4}}$; for readability, let ${w_i =
\weight{\striple_i}}$ and ${s_i = \size{\striple_i}}$ for ${i \in
\setof{1,2}}$. Now every expansion of $\striple_1$ is determined by a choice of
$w_1$ triples from $s_1$ possibilities, so there are a total of
${\binom{s_1}{w_1}}$ expansions of $\striple_1$. Moreover, each expansion of
$\striple_1$ that contains $\gtriple_1$ is obtained by choosing $\gtriple_1$
and the remaining ${w_1-1}$ triples from ${s_1-1}$ possibilities; thus, a total
of ${\binom{s_1-1}{w_1-1}}$ expansions of $\striple_1$ contain $\gtriple_1$.
Analogously, $\striple_2$ has ${\binom{s_2}{w_2}}$ expansions, of which
${\binom{s_2-1}{w_2-1}}$ contain $\gtriple_2$. Now a key observation is that
$\gtriple_1$ and $\gtriple_2$ are summarised as distinct triples; consequently,
the expansions of $\striple_1$ and $\striple_2$ are independent, and so the
total number of expansions of $\striple_1$ and $\striple_2$, and the number of
expansions containing both $\gtriple_1$ and $\gtriple_2$, are given by the
product of the above factors. Finally, the expansions of each triple ${\striple
\in \sgraph \setminus \setof{\striple_1,\striple_2}}$ are independent from and
thus irrelevant to $q_1$, so they do not affect $\expected{q_1}$. Hence, we get
\begin{displaymath}
    \expected{q_1}  = \frac{\binom{s_1-1}{w_1-1}}{\binom{s_1}{w_1}} \cdot \frac{\binom{s_2-1}{w_2-1}}{\binom{s_2}{w_2}}
                    = \frac{w_1}{s_1} \cdot \frac{w_2}{s_2}
                    = \frac{2}{2 \cdot 2} \cdot \frac{2}{2 \cdot 2}
                    = 0.25.
\end{displaymath}

Next, we demonstrate how to handle queries with variables. Let ${q_2 = \setof{
\tuple{x, \mathit{manages}, y}, \tuple{y, \mathit{owns}, z} }}$. For each
${\graph' \in \semantics}$ and each ${\pi \in \ans{q_2}{\graph'}}$, query
$\pi(q_2)$ does not contain variables, so we can compute $\expected{\pi(q_2)}$
just as for $q_1$. We can thus compute $\expected{q_2}$ by summing the
contribution of each such $\pi$, but enumerating all $\pi$ would be
inefficient. However, for each ${\pi \in \ans{q_2}{\graph'}}$, there exists
${\tau \in \ans{\mu(q_2)}{\sgraph}}$ such that ${\mu(\pi(x)) = \tau(x)}$ for
all ${x \in \dom{\pi}}$; we say that $\tau$ \emph{expands into} $\pi$. Thus,
for each $\tau$, we can compute $\expected{\pi(q_2)}$ for just one prototypical
substitution $\pi$ that $\tau$ expands into, and multiply $\expected{\pi(q_2)}$
by the number of the expansions of $\tau$. On our example, evaluating
$\mu(q_2)$ on $\sgraph$ produces substitutions ${\tau_1 = \setof{x \mapsto b_1,
y \mapsto b_3, z \mapsto b_4}}$, ${\tau_2 = \setof{x \mapsto b_1, y \mapsto
b_1, z \mapsto b_2}}$, \,\,\, ${\tau_3 = \setof{x \mapsto b_1, y \mapsto b_3, z
\mapsto b_2}}$. For $\tau_1$, we can select, say, ${\pi = \setof{ x \mapsto
e_1, y \mapsto e_3, z \mapsto c_3 }}$ as a prototypical expansion of $\tau_1$;
then ${\pi(q_2) = q_1}$, and so ${\expected{\pi(q_2)} = 0.25}$. Moreover, since
$\tau_1$ maps both atoms of $q_2$ to distinct triples in $\sgraph$, each
expansion $\pi$ of $\tau_1$ does the same; thus, we can obtain all such $\pi$
by expanding each variable in $\dom{\tau_1}$ independently, and so there are
${\size{\tau_1(x)} \cdot \size{\tau_1(y)} \cdot \size{\tau_1(z)} = 2 \cdot 2
\cdot 2 = 8}$ expansions of $\tau_1$. Thus, $\tau_1$ contributes to
$\expected{q_2}$ by ${8 \cdot 0.25 = 2}$. We compute the contributions of
$\tau_2$ and $\tau_3$ analogously as ${2^3 \cdot \frac{1 \cdot 1}{4 \cdot 4} =
0.5}$ and ${2^3 \cdot \frac{2 \cdot 1}{4 \cdot 4} = 1}$. By summing all
contributions, we get ${\expected{q_2} = 2 + 0.5 + 1 = 3.5}$.

Generalising to any CQ $q$, formula \eqref{eq:freeexpect} sums the contribution
of each ${\tau \in \ans{\mu(q)}{\sgraph}}$: the first factor counts the
expansions of $\tau$ to a prototypical $\pi$, and the second one counts the
contribution of $\pi$.
\begin{equation}
    \expected{q} = \sum_{\tau \in \ans{\mu(q)}{\sgraph}} \ \prod_{x \in \var{q}} \size{\tau(x)} \cdot \prod_{\qatom \in q} \frac{\weight{\tau(\mu(\qatom))}}{\size{\tau(\mu(\qatom))}} \label{eq:freeexpect}
\end{equation}
However, formula \eqref{eq:freeexpect} is correct only if each answer to
$\mu(q)$ in $H$ maps all atoms of $\mu(q)$ to distinct triples in $\sgraph$.
Consider a query ${q_3 = \setof{\tuple{e_3, \mathit{owns}, x}, \tuple{e_3,
\mathit{owns}, y}}}$, so $\ans{\mu(q_3)}{\sgraph}$ contains answers
\begin{displaymath}
\begin{array}{@{}r@{\;}lr@{\;}l@{}}
    \tau_1  & = \setof{x \mapsto b_2, y \mapsto b_4},   & \tau_2    & = \setof{x \mapsto b_4, y \mapsto b_2}, \\
    \tau_3  & = \setof{x \mapsto b_4, y \mapsto b_4},   & \tau_4    & = \setof{x \mapsto b_2, y \mapsto b_2}.
\end{array}
\end{displaymath}
Answers $\tau_1$ and $\tau_2$ map the atoms of $\mu(q_3)$ to distinct triples
in $\sgraph$, so their contributions are given by \eqref{eq:freeexpect} as ${2
\cdot 2 \cdot \frac{2}{2 \cdot 2} \cdot \frac{1}{2 \cdot 2} = 0.5}$. In
contrast, $\tau_3$ maps both atoms of $\mu(q_3)$ to ${\striple_3 = \tuple{b_3,
\mathit{owns}, b_4}}$, so $\tau_3$ expands into two kinds of $\pi$. First,
$\pi$ can map both atoms of $q_3$ to the same triple (e.g., ${\tuple{e_3,
\mathit{owns}, c_4} \in \graph_2}$). Then $\pi(q_3)$ contains just one atom, so
$\expected{\pi(q_3)} = \myfrac{\weight{\striple_3}}{\size{\striple_3}} =
\myfrac{2}{4} = 0.5$; also, $x$ and $y$ are mapped to the same value in each
such $\pi$, so the number of such $\pi$ is ${N_1 = \size{\tau_3(x)} = 2}$.
Second, $\pi$ can map the atoms of $q_3$ to distinct triples. Then, each
expansion of $\striple_3$ containing $\pi(q_3)$ corresponds to a choice of
${\weight{\striple_3}-2}$ triples out of ${\size{\striple_3}-2}$ possibilities,
so
\begin{displaymath}
    \expected{\pi(q_3)} = \frac{\binom{\size{\striple_3}-2}{\weight{\striple_3}-2}}{\binom{\size{\striple_3}}{\weight{\striple_3}}} = \frac{(\weight{\striple_3})_2}{(\size{\striple_3})_2} = \frac{(2)_2}{(4)_2} = \frac{2 \cdot 1}{4 \cdot 3} \approx 0.167.
\end{displaymath}
Moreover, there are ${N_2 = \size{\tau_3(x)} \cdot \size{\tau_3(y)} = 2 \cdot 2
= 4}$ expansions of $\tau_3$, and ${N_1 = 2}$ of these map $x$ and $y$ to the
same resource; thus, ${N_2 - N_1 = 2}$ expansions of $\tau_3$ map $x$ and $y$
to distinct resources; and so $\tau_3$ contributes to $\expected{q_3}$ by ${N_1
\cdot 0.5 + (N_2 - N_1) \cdot 0.167 = 1.33}$. Finally, answer $\tau_4$ maps both atoms
of $\mu(q_3)$ to the same triple ${\striple_4 = \tuple{b_3, \mathit{owns},
b_2}}$, but no expansion of $\tau_4$ contains two triples since
${\weight{\striple_4} = 1}$; thus, we consider only ${\size{\tau_4(x)} = 2}$
expansions $\pi$ of $\tau_4$, each mapping both atoms of $q_3$ to the same
triple and thus contributing by ${\expected{\pi(q_3)} = \myfrac{1}{2} = 0.5}$;
hence, $\tau_4$ contributes to $\expected{q_3}$ by ${2 \cdot 0.5 = 1}$.
Finally, we get ${\expected{q_3} \approx 0.5 + 0.5 + 1.33 + 1 = 3.33}$.

\subsection{Formalisation}\label{sec:estimation:formalisation}

For this section, we fix a consistent summary ${\summary = \summarytuple}$, a
CQ $q$ with ${\res{q} \subseteq \dom{\mu}}$, and an arbitrary total order
$\leq$ on terms where ${\gres \leq x}$ holds for each resource $\gres$ and each
variable $x$.

We first introduce a notion of a partition of $q$. Intuitively, a partition
groups the atoms of $q$ into disjoint sets, and it describes the `type' of
answer to $q$ on graphs in $\semantics$ where all atoms from each such group
are mapped to the same triple.

\begin{definition}\label{def:partition}\em
    A \emph{partition} of $q$ is a set $P$ of mutually disjoint nonempty
    subsets $u_i$ of $q$ such that ${q = \bigcup_{u_i \in P} u_i}$.
    
    The \emph{term graph} $\mathcal{G}_P$ for $P$ contains all terms from
    $\term{q}$ as vertices, and undirected edges between terms $s_1$ and $s_2$,
    $p_1$ and $p_2$, and $o_1$ and $o_2$ for each ${u \in P}$ and all atoms
    ${\tuple{s_1,p_1,o_1}}$ and ${\tuple{s_2,p_2,o_2}}$ in $u$. Partition $P$
    is \emph{unifiable} if no two distinct resources from $\res{q}$ are
    reachable in $\mathcal{G}_P$, in which case substitution $\psubst_P$ maps
    each variable ${x \in \var{q}}$ to the $\leq$-least term reachable from $x$
    in $\mathcal{G}_P$.
    
    The \emph{partition base} $\pbase$ \emph{for} $q$ is the set of all
    unifiable partitions of $q$. The partial order $\preceq$ on $\pbase$ is
    defined so that, for ${P,P' \in \pbase}$, relationship ${P \preceq P'}$
    holds iff, for each ${u \in P}$, there exists ${u' \in P'}$ with ${u
    \subseteq u'}$.

    For ${P,P' \in \pbase}$ with ${P \preceq P'}$, a \emph{chain from} $P$
    \emph{to} $P'$ \emph{of length} ${\ell \geq 0}$ is a sequence ${P = P_0,
    \ldots, P_\ell = P'}$ of partitions from $\pbase$ where ${P_{i-1} \prec
    P_i}$ holds for ${i \in \interval{1}{\ell}}$; moreover, $\Ceven{P}{P'}$ and
    $\Codd{P}{P'}$ are the numbers of chains from $P$ to $P'$ of even and odd
    length, respectively; finally, ${\dcoeff{P}{P'} = \Ceven{P}{P'} -
    \Codd{P}{P'}}$.
\end{definition}

Note that $\emptyset$ is the only unifying partition of ${q = \emptyset}$.
Moreover, ${\Ceven{P}{P} = 1}$ and ${\Codd{P}{P} = 0}$ for each ${P \in
\pbase}$, so ${\dcoeff{P}{P} = 1}$.

Let ${q_3 = \{ \qatom_1, \qatom_2 \}}$ for ${\qatom_1 = \tuple{e_3,
\mathit{owns}, x}}$ and ${\qatom_2 = \tuple{e_3, \mathit{owns}, y}}$ be as in
Section~\ref{sec:estimation:intuition}. Then, ${P_1 = \setof{ \setof{ \qatom_1
}, \setof{ \qatom_2 } }}$ and ${P_2 = \setof{ \setof{ \qatom_1, \qatom_2 } }}$
are the unifiable partitions of $q_3$: partition $P_1$ represents the answers
to $q_3$ on graphs in $\semantics$ that map $\qatom_1$ and $\qatom_2$ to
distinct triples, and $P_2$ represents the answers that map $\qatom_1$ and
$\qatom_2$ to the same triple; the latter is captured by ${\psubst_{P_2} =
\setof{ y \mapsto x}}$ (assuming ${x \leq y}$). Also, ${P_1 \prec P_2}$ and
${\dcoeff{P_1}{P_2} = -1}$ show that the answers to $q_3$ satisfying $P_2$ are
exactly the answers to $q_3$ minus the answers to $q_3$ that satisfy $P_1$.

For an answer $\tau$ to $\mu(q)$ on $\sgraph$ and an appropriate partition $P$,
we define coefficients $\coeff{\tau}{P}$ and $\factor{\tau}{P}$ as follows.

\begin{definition}\label{def:factor}\em
    Let $\tau$ be an answer to $\mu(q)$ on $\sgraph$. A partition ${P \in
    \pbase}$ is \emph{satisfied} by $\tau$ if, for each ${x \in \var{q}}$,
    (i)~${\psubst_P(x) \in \var{q}}$ implies ${\tau(x) = \tau(\psubst_P(x))}$,
    and (ii)~ ${\psubst_P(x) \in \res{q}}$ implies ${\tau(x) =
    \mu(\psubst_P(x))}$. Set $\pbase_\tau$ contains all partitions of $\pbase$
    satisfied by $\tau$. For ${P \in \pbase_\tau}$, let
    \begin{align*}
        \coeff{\tau}{P}     & = \prod_{x \in\vrng{\psubst_P}} \size{\tau(x)} \qquad\qquad \text{and} \\[1ex]
        \factor{\tau}{P}    & = \begin{cases}
                                    \prod\limits_{\striple \in \tau(\mu(q))} \frac{(\weight{\striple})_{\pcount{\striple}}}{(\size{\striple})_{\pcount{\striple}}}  & \text{if } \pcount{\striple} \leq \weight{\striple} ~ \forall \striple \in \tau(\mu(q)) \\[1ex]
                                    0                                                                                                                               & \text{otherwise}
                                 \end{cases} \\
                            & \hspace{1.5cm} \text{where } \pcount{\striple} = |\setof{u \in P \mid \tau(\mu(u)) = \setof{\striple}}|.
    \end{align*}
\end{definition}

In our example, ${\psubst_{P_2}(x) = \psubst_{P_2}(y)}$ and ${\tau_1(x) \neq
\tau_1(y)}$ imply that $P_2$ is not satisfied by $\tau_1$, which intuitively
means that no expansion of $\tau_1$ can match the atoms of $q_3$ as required by
$P_2$. Moreover, ${\coeff{\tau_3}{P_1} = \size{\tau_3(x)} \cdot
\size{\tau_3(y)} = 4}$ and ${\coeff{\tau_3}{P_2} = \size{\tau_3(x)} = 2}$ count
the expansions of $\tau_3$ that comply with $\psubst_{P_1}$ and
$\psubst_{P_2}$, respectively; thus, the number of expansions of $\tau_3$
complying with just $P_1$ is
\begin{displaymath}
    \dcoeff{P_1}{P_1} \cdot \coeff{\tau_3}{P_1} + \dcoeff{P_1}{P_2} \cdot \coeff{\tau_3}{P_2} = 1 \cdot 4 - 1 \cdot 2 = 2.
\end{displaymath}
The contribution of each expansion of $\tau_3$ matching the atoms of $q_3$ to
distinct triples is $\factor{\tau_3}{P_1}$: set $\tau_3(\mu(u))$ contains just
$\striple_3$ for each ${u \in P_1}$, so ${\pcount{\striple_3} = 2}$, and
${\factor{\tau_3}{P_1} \approx 0.167}$.

Thus, to compute $\expected{q}$, we evaluate $\mu(q)$ in $\sgraph$ and, for
each answer $\tau$, we combine $\factor{\tau}{P}$ and $\coeff{\tau}{P'}$ as in
Theorem~\ref{th:expect}. Enumerating all $\tau$ corresponds to query answering
and is thus polynomial in $|\sgraph|$, and the size of $\pbase_\tau$ does not
depend on $\sgraph$; hence $\expected{q}$ can be computed in time polynomial in
$|\sgraph|$. Moreover, Theorem~\ref{th:variance} shows how to compute the
variance of the cardinality of $q$.

\begin{restatable}{theorem}{expect}\label{th:expect}
    The following identity holds:
    \begin{align*} 
        \expected{q} = \sum_{\tau \in \ans{\mu(q)}{\sgraph}} \ \sum_{P \in \pbase_\tau} \factor{\tau}{P} \cdot \sum_{\substack{P' \in \pbase_\tau, P \preceq P'}} \dcoeff{P}{P'} \cdot \coeff{\tau}{P'}.
    \end{align*} 
\end{restatable}

\begin{restatable}{theorem}{variance}\label{th:variance}
    Identity ${\std{q}^2 = \expected{q \cup \rho(q)} - \expected{q}^2}$ holds,
    where $\rho$ is a substitution mapping each variable in $\var{q}$ to a
    fresh variable.
\end{restatable}

We next introduce the notion of \emph{$\mu$-unification-free queries}, where no
two distinct atoms can ever be mapped to the same triple.
Proposition~\ref{prop:free-expect} shows that, for such queries, the simpler
formula \eqref{eq:freeexpect} correctly computes the expectation $\expected{q}$.

\begin{definition}\label{def:unification-free}\em
    Atoms $\qatom_1$ and $\qatom_2$ are \emph{unifiable} if a substitution
    $\kappa$ exists such that ${\kappa(\qatom_1) = \kappa(\qatom_2)}$. Query
    $q$ is \emph{$\mu$-unification-free} if no two distinct atoms of $\mu(q)$
    are unifiable; otherwise, query $q$ is \emph{$\mu$-unifiable}.
\end{definition}

\begin{restatable}{proposition}{freeexpect}\label{prop:free-expect}
    Formula~\eqref{eq:freeexpect} correctly computes $\expected{q}$ when $q$ is
    a $\mu$-unification-free query.
\end{restatable}

Queries $q_1$ and $q_2$ from Section~\ref{sec:estimation:intuition} are
$\mu$-unification-free, while $q_3$ is not. Moreover, query ${q_4 =
\setof{\tuple{e_3, \mathit{owns}, x}, \tuple{e_4, \mathit{owns}, y}}}$ is also
\emph{not} $\mu$-unification-free: its atoms unify \emph{after} applying $\mu$
to $q_4$. Note that queries used in practice often consist of atoms of the form
$\tuple{t, r, t'}$ and $\tuple{t, \rdftype, r}$ where $t$ and $t'$ are terms
and $r$ is a distinct resource; then, such queries are unification-free
whenever each such resource $r$ is assigned to a distinct bucket.

If $q$ is $\mu$-unification-free, then ${P = \setof{ \setof{ \qatom_i } \mid
\qatom_i \in q}}$ is the only unifiable partition of $q$, so the estimation
formula simplifies to \eqref{eq:freeexpect}. In particular, $P'$ in the last
sum in Theorem~\ref{th:expect} can in such cases only be $P$, so the summation
reduces to $\coeff{\tau}{P}$---the first product of \eqref{eq:freeexpect}.
Moreover, the atoms of $\mu(q)$ cannot be equated, so ${\pcount{\striple} = 1}$
for each ${\striple \in \tau(\mu(q))}$; thus, $\factor{\tau}{P}$ reduces to the
second product of \eqref{eq:freeexpect}.

\section{Summarising RDF Graphs}\label{sec:summarisation}

As we have mentioned in Section~\ref{sec:introduction}, reusing existing graph
summarisation techniques proved challenging: some produced summaries that did
not lead to accurate cardinality estimates, and others could not be efficiently
applied to large graphs. We thus developed a new summarisation approach
consisting of two steps we describe next.

\subsection{Typed Summary}\label{sec:summarisation:typed-summary}

To obtain a coarse measure of resource similarity, we assign to each resource
$\gres$ in an RDF graph $\graph$ a \emph{type} ${\type{\gres} =
\tuple{C(\gres), O(\gres), I(\gres), P(\gres)}}$, which is a tuple consisting
of four components described next.

The most important component of $\type{\gres}$ is the \emph{class type}
$C(\gres)$ of $\gres$, which is the set defined as follows:
\begin{displaymath}
    C(\gres) = \begin{cases}
                    \setof{o \in \res{\graph} \mid \tuple{\gres, \rdftype, o} \in \graph}   & \text{if } \gres \text{ is a URI}, \\
                    \setof{ \text{the datatype of } \gres}                                  & \text{if } \gres \text{ is a literal}. \\
               \end{cases}
\end{displaymath}
One can intuitively expect that a summary whose buckets contain, say, both
people and cars, is unlikely to provide good cardinality estimates. Thus, our
summarisation approaches will put resources $\gres_1$ and $\gres_2$ into the
same bucket only if ${C(\gres_1) = C(\gres_2)}$ holds. A similar idea has been
used in query answering over ontologies \cite{DBLP:conf/aaai/DolbyFKKSSM07}.

Sets $O(\gres)$ and $I(\gres)$ describe the outgoing and incoming connections
of $\gres$. Let ${\tuple{p_1, \dots, p_k}}$ be an arbitrary ordering of the
predicates in $\graph$ different from $\rdftype$; then, we could capture the
outgoing connections of $\gres$ by a vector ${\tuple{n_1, \dots, n_k}}$, where
each $n_i$ is the number of the outgoing connections from $\gres$ for predicate
$p_i$. This, however, is likely to create too many types: resources in an RDF
graph are unlikely to have exactly the same outgoing connections with the same
cardinalities. We thus represent the cardinality information using histograms.
For each predicate $p_i$, we compute the set $\mathit{out}_{p_i}$ that, for
each resource $\gres_j$ occurring in $\graph$ in subject position, contains a
pair ${\tuple{\gres_j, \, |\setof{o \in \res{\graph} \text{ such that }
\tuple{\gres_j, p_i, o} \in \graph}| \, }}$ associating $\gres_j$ with the
number of resources that $\gres_j$ is connected to via property $p_i$. We sort
the pairs in $\mathit{out}_{p_i}$ by the value of their second component into a
list ${\tuple{\gres_1, n_1}, \dots, \tuple{\gres_N, n_N}}$. Finally, given a
predetermined number $J_i$ of histogram buckets (not to be confused with
summary buckets) for $p_i$, we create a histogram for $\mathit{out}_{p_i}$: for
${w = \lceil N / J_i \rceil}$, the first bucket contains resources ${\gres_1,
\dots, \gres_w}$ and has ID 1, the second bucket contains resources
${\gres_{w+1}, \dots, \gres_{2w}}$ and has ID 2, and so on. Thus, each bucket
contains resources with similar outgoing frequencies for $p_i$, and the
difference between bucket IDs is indicative of the bucket similarity. Then, the
\emph{outgoing type} $O(\gres)$ of a resource $\gres$ is the vector
$\tuple{O_1, \dots, O_k}$, where each $O_i$ is the ID of the bucket that
$\gres$ was assigned to in the histogram for $p_i$. The \emph{incoming type}
$I(\gres)$ of $\gres$ is defined analogously. The numbers of buckets can be
selected independently for each predicate and direction: it is reasonable to
use a larger number of buckets if the relevant frequencies vary significantly
across resources.

Finally, we also identify highly connected groups of resources: one can
intuitively expect that such resources should be assigned to the same bucket.
To this end, we divide the resources into a predetermined number of partitions
of roughly the same size while minimising the number of edges between
partitions using the algorithms by \citet{Karypis:1998}. The \emph{partition
type} $P(\gres)$ of a resource $\gres$ is the partition to which $\gres$ was
assigned to.

Thus, if the types $\type{\gres_1}$ and $\type{\gres_2}$ of resources $\gres_1$
and $\gres_2$ are the same, then $\gres_1$ and $\gres_2$ occur in the same
classes (or datatypes), have comparable outgoing and incoming frequencies for
all predicates, and belong to the same highly connected component of the RDF
graph; hence, we can expect $\gres_1$ and $\gres_2$ to be similar. Thus, we
define the \emph{typed summary} ${\summary = \tuple{\sgraph, \weightfunction,
\mu}}$ of an RDF $\graph$ as follows: for each resource ${\gres \in
\res{\graph}}$, if $\gres$ occurs in $\graph$ in triples of the form $\tuple{s,
\gres, o}$ or $\tuple{s, \rdftype, \gres}$, we let ${\mu(\gres) = \gres}$;
otherwise, we let ${\mu(\gres) = b_{\type{\gres}}}$, where $b_{\type{\gres}}$
is a distinct bucket uniquely associated with the type $\type{\gres}$ of
$\gres$. Finally, we define ${\sgraph = \mu(\graph)}$ and adjust the function
$\weightfunction$ accordingly.

\subsection{Summary Refinement by MinHashing}\label{sec:summarisation:MinHash-summary}

If a typed summary is large, we compress it by merging similar buckets and
types. To identify similar buckets, we define the \emph{vicinity} ${\vic(b)}$
of a bucket ${b \in \res{\sgraph}}$ as the union of the following sets:
\begin{align*}
    \vico(b)    & = \setof{\tuple{b_p, b_o}  \mid  \tuple{b, b_p, b_o} \in \sgraph \wedge b_p \neq \mu(\rdftype)}, \\
    \vici(b)    & = \setof{\tuple{b_i, b_p}  \mid \tuple{b_i, b_p, b} \in \sgraph \wedge b_p \neq \mu(\rdftype)}.
\end{align*}
Since vicinity $\vic(b)$ describes the connections of $b$, given buckets
${b_1,b_2 \in \res{\sgraph}}$, the degree of commonality between $\vic(b_1)$
and $\vic(b_2)$ is indicative of the similarity of $b_1$ and $b_2$. This is
captured by the \emph{Jaccard index} \cite{DBLP:books/cu/LeskovecRU14} of $b_1$
and $b_2$, which is defined as
\begin{displaymath}
    \Jind{\summary}{b_1}{b_2} = \frac{|\vic(b_1) \cap \vic(b_2)|}{|\vic(b_1) \cup \vic(b_2)|}
\end{displaymath}
if ${\vic(b_1) \cup \vic(b_2) \neq \emptyset}$, and otherwise
${\Jind{\summary}{b_1}{b_2} = 1}$. Thus, we can na{\"i}vely reduce the size of
$\summary$ by repeatedly merging the pair of buckets $b_1$ and $b_2$ with
maximal $\Jind{\summary}{b_1}{b_2}$, but this is impractical for two reasons.
First, computing $\Jind{\summary}{b_1}{b_2}$ requires iterating over
$\vic(b_1)$ and $\vic(b_2)$, which can be large. Second, the number of pairs of
$b_1$ and $b_2$ can be large in even moderately sized summaries.

We address the first problems using MinHashing
\cite{DBLP:books/cu/LeskovecRU14}, which uses  fixed-sized signatures to
approximate $\Jind{\summary}{b_1}{b_2}$. Given two integer parameters $\mdim$
and $\ndim$, we generate a \emph{MinHash scheme $\MHscheme$ of size ${\mdim
\times \ndim}$}, which is an ${\mdim \times \ndim}$ matrix of hash functions
chosen uniformly at random with replacement. Each $\MHscheme[i,j]$ maps
elements of $\vic(b)$ (i.e., pairs of resources) to natural numbers. The
\emph{signature} of a bucket ${b \in \res{\sgraph}}$ on $\MHscheme$ and
$\summary$ is the ${\mdim \times \ndim}$ matrix $\MHsig{b}$ where
\begin{displaymath}
    \MHsig{b}[i,j] = \min_{\alpha \in \vic(b)} \MHscheme[i,j](\alpha)
\end{displaymath}
if ${\vic(b) \neq \emptyset}$, and otherwise ${\MHsig{b}[i,j] = \infty}$. It is
known that
\begin{displaymath}
    \appJind{\summary}{b_1}{b_2} = \dfrac{| \setof{ \tuple{i,j} \mid  \MHsig{b_1}[i,j] = \MHsig{b_2}[i,j] }|}{\mdim \cdot \ndim} \\
\end{displaymath}
is a good approximation of the Jaccard index of buckets $b_1$ and $b_2$
\cite{DBLP:books/cu/LeskovecRU14}. For ${i \in \interval{1}{\mdim}}$, we write
$\MHsig{b}[i]$ for the $i$-th row of $\MHsig{b}$.

We address the second problem using locality sensitive hashing
\cite{DBLP:books/cu/LeskovecRU14}: we generate a hash function $\Call{lsh}{}$
for $\MHsig{b}[i]$ and use it to assign each bucket $b$ to a bin
$\binof{\Call{lsh}{\MHsig{b}[i]}}$. For each ${\bin \in \bins}$ and all buckets
${b_1,b_2 \in \bin}$, it is known that $\appJind{\summary}{b_1}{b_2}$ is likely
to be high, so all pairs of buckets in $\bin$ are likely to be
similar~\cite{DBLP:books/cu/LeskovecRU14}.

\begin{algorithm}[t]
\caption{MinHash summarisation algorithm}\label{algo:summ}
\begin{algorithmic}[1]
    \algrenewcommand\algorithmicindent{0.35cm}
    \Statex
    \begin{tabular}{@{}l@{\quad}l@{\;}c@{\;}l@{}}
    \textbf{Inputs:}    & $\graph$              & : & an RDF graph \\
                        & $\types$              & : & a set of types \\
                        & $\mdim \times \ndim$  & : & size of the MinHash scheme \\
                        & $\target$             & : & target size of the summary \\
    \end{tabular}
    \Statex
    \State Generate a MinHash scheme $\MHscheme$ of size $\mdim \times \ndim$
    \State Generate a locality sensitive hash $\Call{lsh}{} : \nat^\ndim \to \nat$
    \State Let $\summary \defeq \summarytuple$ be the trivial summary of $\graph$                               \label{algo:summ:create-trivial}
    \ForAll{$\type{} \in \types$}
        $\buckets_{\type{}} \defeq \{ b \in \res{\sgraph} \mid b \text{ is of type } \type{} \}$                \label{algo:summ:create-buckets}
    \EndFor
    \Loop                                                                                                       \label{algo:summ:loop:start}
        \If{$|\sgraph| \leq \target$}
            \Return $\summary$                                                                                  \label{algo:summ:success}
        \EndIf
        \State $\mergeQueue \defeq \emptyset$
        \ForAll{$\type{} \in \types$}                                                                           \label{algo:summ:t:start}
            \ForAll{$b \in \buckets_{\type{}}$}
                Compute $\MHsig{b}$ using $\MHscheme$ and $\summary$                                            \label{algo:summ:signature}
            \EndFor
            \For{$i \defeq 1$ \textbf{to} $\mdim$}                                                              \label{algo:summ:row:start}
                \State $\bins \defeq \emptyset$
                \ForAll{$b \in \buckets_{\type{}}$}
                    Add $b$ to $\binof{\Call{lsh}{\MHsig{b}[i]}}$                                               \label{algo:summ:add-to-bin}
                \EndFor
                \For{$\bin \in \bins$ such that $|\bin| \geq 2$}                                                \label{algo:summ:merge:start}
                    \State Choose some $b \in \bin$
                    \For{$b' \in \bin$ with $b \neq b'$}
                        \State Add $\tuple{b',b}$ to $\mergeQueue$ and remove $b'$ from $\buckets_{\type{}}$
                    \EndFor
                \EndFor                                                                                         \label{algo:summ:merge:end}
            \EndFor                                                                                             \label{algo:summ:row:end}
        \EndFor                                                                                                 \label{algo:summ:t:end}
        \State $\sizebefore \defeq |\sgraph|$
        \ForAll{$\tuple{b', b} \in \mergeQueue$}
            Merge $b'$ into $b$ in $\summary$                                                                   \label{algo:summ:apply-merge}
        \EndFor
        \If{$\sizebefore - |\sgraph| \leq (|\sgraph| - \target) \cdot 0.01$}
            \State Merge similar types in $\types$                                                              \label{algo:summ:merge-sim}
            \If{no merge has happened}
                \Return $\summary$
            \EndIf
        \EndIf
    \EndLoop                                                                                                    \label{algo:summ:loop:end}
\end{algorithmic}
\end{algorithm}

Algorithm~\ref{algo:summ} uses these ideas to compute a summary of a graph
$\graph$. It takes as arguments a set of types $\types$ covering all resources
of $\graph$. The algorithm first converts $\graph$ into a trivial summary
$\summary$ (line~\ref{algo:summ:create-trivial}) where ${\sgraph = \graph}$ and
${\mu(\gres) = \gres}$ for each ${\gres \in \res{\graph}}$, and then it
computes the sets $\buckets_{\type{}}$ of buckets of type $\type{}$ for each
${\type{} \in \types}$ (line~\ref{algo:summ:create-buckets}). The algorithm
next enters a loop (lines~\ref{algo:summ:loop:start}--\ref{algo:summ:loop:end})
that merges buckets of the same type. It returns $\summary$ if the number of
edges in $\sgraph$ is in the required space budget
(line~\ref{algo:summ:success}); otherwise, it considers each type ${\type{} \in
\types}$ (lines~\ref{algo:summ:t:start}--\ref{algo:summ:t:end}) and computes
the signature of each bucket ${b \in \buckets_{\type{}}}$
(line~\ref{algo:summ:signature}). For each row $i$ of the signature
(lines~\ref{algo:summ:row:start}--\ref{algo:summ:row:end}), the algorithm
applies locality sensitive hashing to each bucket ${b \in \buckets_{\type{}}}$
(line~\ref{algo:summ:add-to-bin}) and adds $b$ to a bin. The number of bins is
not predetermined, which reduces collisions. Finally, for all bins containing
at least two buckets, all buckets in the bin are scheduled to be merged into
one designated bucket $b$
(lines~\ref{algo:summ:merge:start}--\ref{algo:summ:merge:end}). Locality
sensitive hashing ensures that similar buckets are likely to end in the same
bin, but this is not guaranteed so the standard algorithm would merge $b'$ and
$b$ only if the estimate $\appJind{\summary}{b}{b'}$ is above a particular
threshold. We skip this step for two reasons: the quadratic step of computing
$\appJind{\summary}{b}{b'}$ for all pairs of buckets in a bin can be
prohibitive, and the chance that dissimilar buckets are assigned to the same
bin is small since we do not limit the number of bins. After all signature rows
have been considered, the merges are applied to the summary
(line~\ref{algo:summ:apply-merge}). Finally, if the size of $\summary$ has not
been reduced sufficiently, then the types in $\types$ are too specific to allow
further merges, so the algorithm merges similar types in $\types$
(line~\ref{algo:summ:merge-sim}) to allow further reducing in the size of
$\summary$.

Let $\type{}$ and $\type{}'$ be types whose outgoing and incoming types are
$\tuple{O_1, \dots, O_k}$ and $\tuple{O_1', \dots, O_k'}$, and $\tuple{I_1,
\dots, I_k}$ and $\tuple{I_1', \dots, I_k'}$. Then, $\type{}$ and $\type{}'$
can be similar only if their class and partition types coincide, and their
similarity $\mathcal{J}(\type{},\type{}')$ is the average of the generalised
Jaccard indexes \cite{DBLP:books/cu/LeskovecRU14} of their outgoing and
incoming types:
\begin{displaymath}
    \mathcal{J}(\type{},\type{}') = \frac{1}{2}\left(\sum_i\frac{\min(O_i,O_i')}{\max(O_i,O_i')} + \sum_i\frac{\min(I_i,I_i')}{\max(I_i,I_i')}\right).
\end{displaymath}
Merging $\type{}$ and $\type{}'$ produces $\type{}''$ with the class and
partition types same as $\type{}$ and $\type{}'$, outgoing type ${\tuple{(O_1 +
O_1')/2, \dots, (O_k + O_k')/2}}$, and an analogous incoming type; we also set
${\buckets_{\type{}''} = \buckets_{\type{}} \cup \buckets_{\type{}'}}$.

To merge the types in line~\ref{algo:summ:merge-sim}, the algorithm first
groups the types by their class and partition type. Ideally, we would merge the
most similar pairs within each such group, but computing the similarity of each
pair in the group would be prohibitively expensive. Instead, the algorithm
selects 500 pairs at random, computes their similarity, and merges the most
similar pair if their similarity is at least 50\%. This process is repeated
until the size of $\types$ is reduced by 20\% (if possible), which then allows
for further merging of buckets.

\subsection{Choosing the Parameters}

To recapitulate, our algorithm is parametrised by (i)~the number of buckets for
the incoming and outgoing histograms per predicate, (ii)~the number of
partitions, (iii)~the target size of the summary, and (iv)~the size of the
MinHash scheme ${m \times n}$ (if used). Standard techniques can be used to
choose the number of buckets, such as the Freedman--Diaconis rule
\cite{Freedman1981}. We found it reasonable to partition a graph into 10
partitions, and fall back to just one partition if more than 20\% of the edges
are cut. We also identified ${\mdim = 20}$, ${\ndim = 2}$, and target sizes in
the order of tens of thousands as reasonable defaults.

\section{Experimental Evaluation}\label{sec:evaluation}

\begin{table*}[p]
    \centering
    \footnotesize
    \setlength{\tabcolsep}{2pt}
    \def\arraystretch{1.2}

    \captionof{table}{Dataset, summary, and query statistics}\label{table:data}
    \begin{tabular}{|c|*{22}{c|}}
                                \cline{1-23}
        \multicolumn{1}{|c|}{}  & \multicolumn{2}{c|}{Original graph $\graph$}  & \multicolumn{4}{c|}{Summary $\summary$}               & \multicolumn{16}{c|}{Queries by type} \\
                                \cline{2-23}
        \multicolumn{1}{|c|}{Benchmark} & Nr            & Nr                    & Nr        & Nr        & Reduction & Construction      & \multicolumn{4}{c|}{Linear}           & \multicolumn{4}{c|}{Star}             & \multicolumn{4}{c|}{Snowflake}            & \multicolumn{4}{c|}{Complex} \\
                                                                                                                                                    \cline{8-23}
        \multicolumn{1}{|c|}{}          & resources     & triples               & buckets   & triples   & factor    & time              & total & $\mu$-un.\    & min   & max   & total & $\mu$-un.\    & min   & max   & total & $\mu$-un.\    & min   & max       & total & $\mu$-un.\    & min  & max \\
        \hline
        LUBM                            & 16.4~M        & \hspace{2pt} 91.1~M   & 313       & 14,926    & 6.1~k     & 6~min             & 16    & 3             & 1     & 6     & 2     & 0             & 4     & 10    & 6     & 0             & 4     & 10        & 14    & 2             & 4    & 9 \\
        WatDiv                          & 10.2~M        & 108.9~M               & 5,504     & 43,350    & 2.5~k     & 2~h\phantom{in}   & 30    & 0             & 2     & 3     & 40    & 0             & 2     & 9     & 30    & 0             & 5     & 9         & 3     & 0             & 6    & 10 \\
        DBLP                            & 25.4~M        & \hspace{2pt} 55.5~M   & 4,765     & 28,631    & 1.9~k     & 45~min            & 4     & 0             & 1     & 4     & 1     & 1             & 3     & 3     & 2     & 1             & 6     & 7         & 8     & 6             & 5    & 11 \\
        \hline
    \end{tabular}

    \vspace{0.3cm}

    \captionof{table}{Estimates confidence for varying q-errors}\label{table:std}
    \begin{tabular}{|r|*{15}{r|}}
        \hline
        Q-Error $\varepsilon$                                                                   & \multicolumn{5}{c|}{10}                               & \multicolumn{5}{c|}{100}                              & \multicolumn{5}{c|}{1000}           \\
        \hline
        The \% bound on $\prob{\qerror \geq \varepsilon}$ from Theorem~\ref{th:qerrorbound}   & [0,1]     & (1, 5]    & (5, 10] & (10, 25] & $>$ 25   & [0,1]     & (1, 5]    & (5, 10] & (10, 25] & $>$ 25   & [0,1]     & (1, 5]    & (5, 10] & (10, 25] & $>$ 25 \\
        \hline
        LUBM (29)                                                                               & 20 (2)    & 0 (0)     & 6 (1)   & 2 (0)    & 1 (0)    & 25 (3)    & 0  (0)    & 4 (0)   & 0 (0)    & 0  (0)   & 28 (1)    & 1 (0)     & 0 (0)   & 0 (0)    & 0 (0)  \\
        WatDiv (48)                                                                             & 16 (1)    & 6 (1)     & 1 (0)   & 10 (2)   & 15 (3)   & 42 (0)    & 1 (0)     & 0 (0)   & 0 (0)    & 5 (0)    & 43 (0)    & 0 (0)     & 0 (0)   & 0 (0)    & 5 (0)  \\
        DBLP (3)                                                                                & 2 (0)     & 1 (0)     & 0 (0)   & 0 (0)    & 0 (0)    & 3 (0)     & 0 (0)     & 0 (0)   & 0 (0)    & 0 (0)    & 3 (0)     & 0 (0)     & 0 (0)   & 0 (0)    & 0 (0)  \\
        \hline
    \end{tabular}

    \vspace{0.8cm}
    
    \newcommand{\row}[2]{ \rotatebox[origin=l]{90}{\bf #1} & \includegraphics[width=0.64\linewidth,height=4.5cm]{figures/#2} & \begin{minipage}[t]{0.3\linewidth} \input{stats/#2-stats} \end{minipage}}
    \begin{tabular}{@{}l@{\;\;}c@{\quad}c@{}}
         \rotatebox[origin=l]{90}{\bf \hspace{2cm} (a) LUBM} & \includegraphics[width=0.64\linewidth,height=4.5cm]{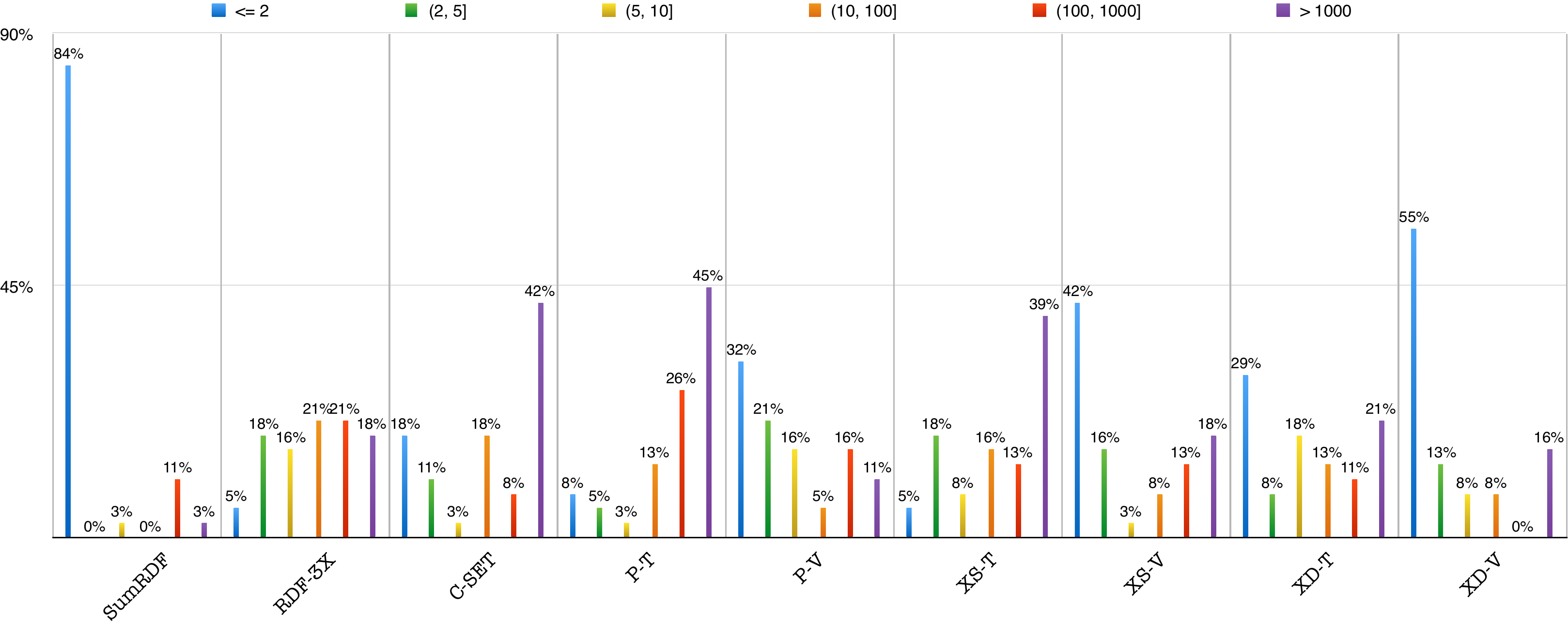} & \begin{minipage}[t]{0.3\linewidth} \usetikzlibrary{backgrounds}
\begin{tikzpicture}
    \begin{axis}[
        width = 1.1\linewidth,
        height= 140pt,
        boxplot/draw direction=y,
        x axis line style={opacity=0},
        axis x line*=bottom,
        axis y line=left,
        font = \tiny,
        ymajorgrids,
        xtick={1,2,3,4,5,6,7,8,9},
        xticklabel style={rotate=45},
        xtick style={draw=none},
        xticklabels={\our, \rdfx, \cset, \postgres{T}, \postgres{V}, \system{S}{T}, \system{S}{V}, \system{D}{T}, \system{D}{V}},
        ytick={1,10,100,1000,10000, 100000},
        ymode = log,
        ymax = 3000000,
        xmin = 0.5,
        xmax = 9.5,
        yminorticks=false,
        yticklabel style={xshift=4pt},
        trim axis right,
     ]

        \addplot[
            black,
            boxplot prepared = {
                lower whisker = 1.00,
                lower quartile = 1.02,
                median = 1.18,
                average = 125.24,
                upper quartile = 1.41,
                upper whisker = 1020.74,
                box extend = 0.8,
            }
        ] coordinates {};

        \addplot[
            black,
            boxplot prepared = {
                lower whisker = 1.00,
                lower quartile = 5.49,
                median = 27.70,
                average = 4552.17,
                upper quartile = 405.15,
                upper whisker = 111055.00,
                box extend = 0.8,
            }
        ] coordinates {};

        \addplot[
            black,
            boxplot prepared = {
                lower whisker = 1.00,
                lower quartile = 3.21,
                median = 105.08,
                average = 36197.10,
                upper quartile = 5341.00,
                upper whisker = 763852.95,
                box extend = 0.8,
            }
        ] coordinates {};

        \addplot[
            black,
            boxplot prepared = {
                lower whisker = 1.37,
                lower quartile = 51.86,
                median = 344.23,
                average =  6e+05,
                upper quartile = 88523.85,
                upper whisker = 2e+06,
                box extend = 0.8,
            }
        ] coordinates {};

        \draw (axis cs: 4.2, 1.9e+06) node[right] {$10^{12}$};
        \draw (axis cs: 4.2, 6e+05) node[right] {$10^{11}$};
        \draw (axis cs: 4, 2e+05) node[red] {$\approx$};

        \addplot[
            black,
            boxplot prepared = {
                lower whisker = 1.00,
                lower quartile = 1.67,
                median = 4.10,
                average = 3244.32,
                upper quartile = 95.73,
                upper whisker = 111055.00,
                box extend = 0.8,
            }
        ] coordinates {};

        \addplot[
            black,
            boxplot prepared = {
                lower whisker = 1.00,
                lower quartile = 7.00,
                median = 174.50,
                average = 6e+05, 
                upper quartile = 9661.98,
                upper whisker =  2e+06,   
                box extend = 0.8,
            }
        ] coordinates {};

        \draw (axis cs: 6, 2e+05) node[red] {$\approx$};
        \draw (axis cs: 6.2, 2e+06) node[right] {$10^8$};
        \draw (axis cs: 6.2, 6e+05) node[right] {$10^7$};

        \addplot[
            black,
            boxplot prepared = {
                lower whisker = 1.00,
                lower quartile = 1.30,
                median = 3.13,
                average = 3566.57,
                upper quartile = 8.62,
                upper whisker = 111055.00,
                box extend = 0.8,
            }
        ] coordinates {};

        \addplot[
            black,
            boxplot prepared = {
                lower whisker = 1.00,
                lower quartile = 1.72,
                median = 7.93,
                average = 5938.44,
                upper quartile = 417.97,
                upper whisker = 136165.00,
                box extend = 0.8,
            }
        ] coordinates {};

        \addplot[
            black,
            boxplot prepared = {
                lower whisker = 1.00,
                lower quartile = 1.03,
                median = 1.76,
                average = 3645.68,
                upper quartile = 8.17,
                upper whisker = 111055.00,
                box extend = 0.8,
            }
        ] coordinates {};

    \end{axis}
 \end{tikzpicture} \end{minipage} \\[2ex]
         \rotatebox[origin=l]{90}{\bf \hspace{2cm} (b) WatDiv} & \includegraphics[width=0.64\linewidth,height=4.5cm]{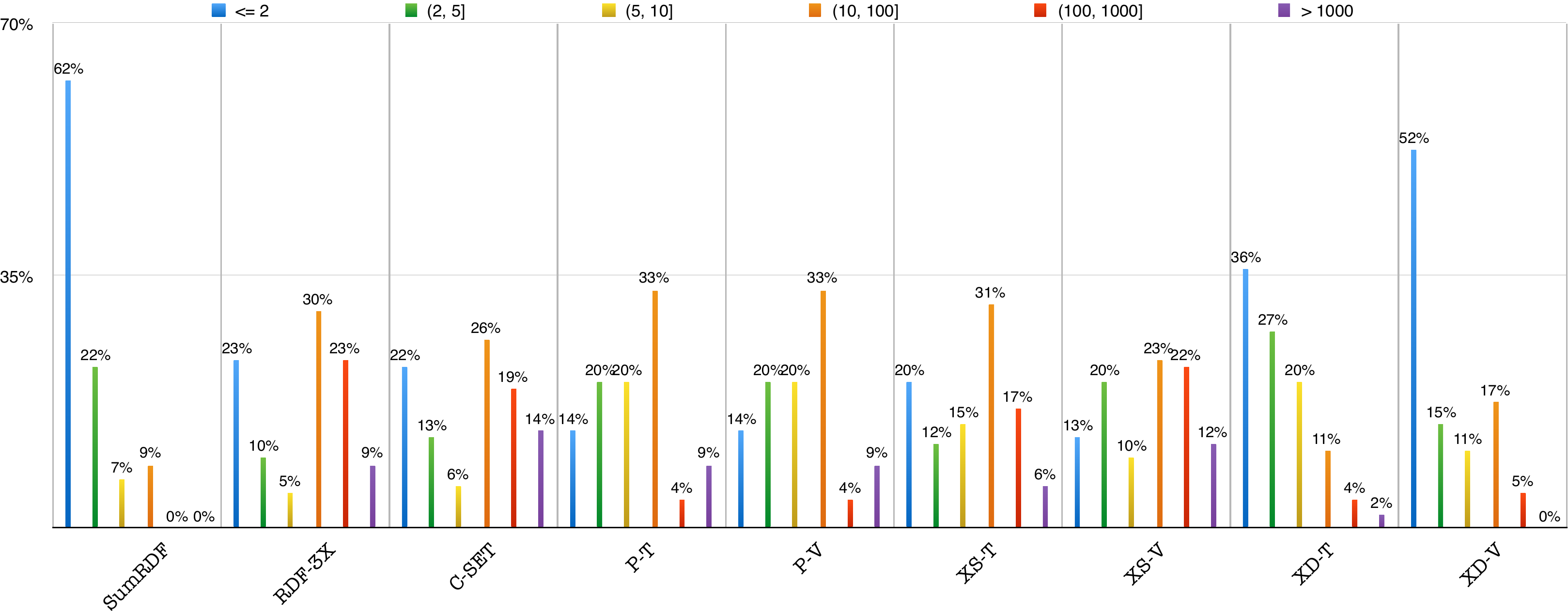} & \begin{minipage}[t]{0.3\linewidth} \begin{tikzpicture}
    \begin{axis}[
        width = 1.1\linewidth,
        height= 140pt,
        boxplot/draw direction=y,
        x axis line style={opacity=0},
        axis x line*=bottom,
        axis y line=left,
        font=\tiny,
        ymajorgrids,
        xtick={1,2,3,4,5,6,7,8,9},
        xticklabel style={rotate=45},
        xtick style={draw=none},
        xticklabels={\our, \rdfx, \cset, \postgres{T}, \postgres{V}, \system{S}{T}, \system{S}{V}, \system{D}{T}, \system{D}{V}},
        ytick={1,10,100,1000,10000, 100000},
        ymode = log,
        ymax = 6000000,
        xmin = 0.5,
        xmax = 9.5,
        yminorticks=false,
        yticklabel style={xshift=4pt}
     ]
        \addplot[
            black,
            boxplot prepared = {
                lower whisker = 1.00,
                lower quartile = 1.03,
                median = 1.37,
                average = 3.14,
                upper quartile = 2.48,
                upper whisker = 18.97,
                box extend = 0.8,
            }
        ] coordinates {};

        \addplot[
            black,
            boxplot prepared = {
                lower whisker = 1.00,
                lower quartile = 2.77,
                median = 26.00,
                average = 6922.33,
                upper quartile = 120.00,
                upper whisker = 508463.00,
                box extend = 0.8,
            }
        ] coordinates {};

        \addplot[
            black,
            boxplot prepared = {
                lower whisker = 1.00,
                lower quartile = 4.13,
                median = 19.00,
                average = 9087.41,
                upper quartile = 210.00,
                upper whisker = 508463.00,
                box extend = 0.8,
            }
        ] coordinates {};

        \addplot[
            black,
            boxplot prepared = {
                lower whisker = 1.0,
                lower quartile = 3.62,
                median = 15.32,
                average = 1e+06, 
                upper quartile = 80.80,
                upper whisker = 4e+06, 
                box extend = 0.8,
            }
        ] coordinates {};

        \draw (axis cs: 4.2, 4e+06) node[right] {$10^8$};
        \draw (axis cs: 4.2, 1e+06) node[right] {$10^6$};

        \draw (axis cs: 4, 508463) node[red] {$\approx$};

        \addplot[
            black,
            boxplot prepared = {
                lower whisker = 1.01,
                lower quartile = 3.29,
                median = 9.00,
                average = 1617.52,
                upper quartile = 35.45,
                upper whisker = 34218.00,
                box extend = 0.8,
            }
        ] coordinates {};

        \addplot[
            black,
            boxplot prepared = {
                lower whisker = 2.88,
                lower quartile = 7.00,
                median = 11.00,
                average = 450.15,
                upper quartile = 64.69,
                upper whisker = 18424.73,
                box extend = 0.8,
            }
        ] coordinates {};

        \addplot[
            black,
            boxplot prepared = {
                lower whisker = 1.01,
                lower quartile = 4.19,
                median = 18.63,
                average = 1491.49,
                upper quartile = 254.88,
                upper whisker = 31443.00,
                box extend = 0.8,
            }
        ] coordinates {};

        \addplot[
            black,
            boxplot prepared = {
                lower whisker = 1.00,
                lower quartile = 1.43,
                median = 3.95,
                average = 129.92,
                upper quartile = 7.10,
                upper whisker = 8851.00,
                box extend = 0.8,
            }
        ] coordinates {};

        \addplot[
            black,
            boxplot prepared = {
                lower whisker = 1.00,
                lower quartile = 1.02,
                median = 2.00,
                average = 16.04,
                upper quartile = 7.65,
                upper whisker = 511.91,
                box extend = 0.8,
            }
        ] coordinates {};

    \end{axis}
\end{tikzpicture} \end{minipage} \\[2ex]
         \rotatebox[origin=l]{90}{\bf \hspace{2cm} (c) DBLP} & \includegraphics[width=0.64\linewidth,height=4.5cm]{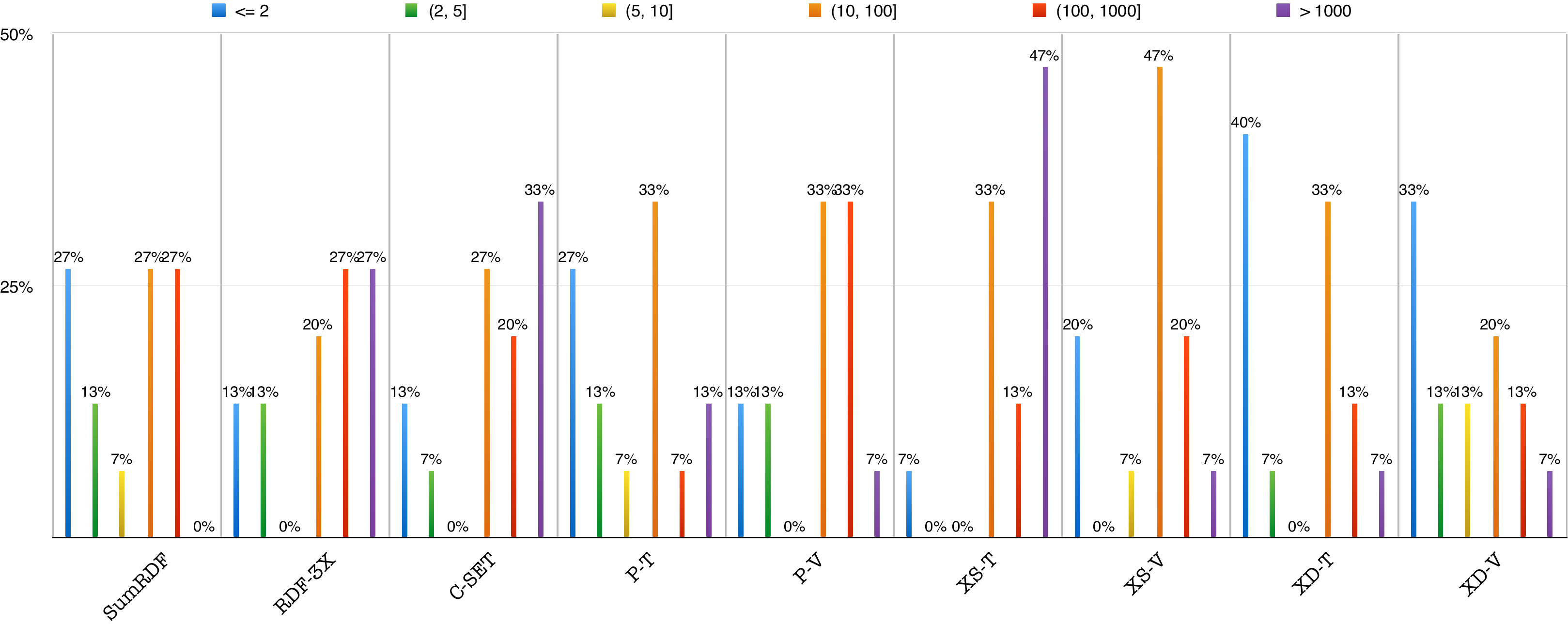} & \begin{minipage}[t]{0.3\linewidth} \begin{tikzpicture}
    \begin{axis}[
        width = 1.1\linewidth,
        height= 140pt,
        boxplot/draw direction=y,
        x axis line style={opacity=0},
        axis x line*=bottom,
        axis y line=left,
        font = \tiny,
        ymajorgrids,
        xtick={1,2,3,4,5,6,7,8,9},
        xticklabel style={rotate=45},
        xtick style={draw=none},
        xticklabels={\our, \rdfx, \cset, \postgres{T}, \postgres{V}, \system{S}{T}, \system{S}{V}, \system{D}{T}, \system{D}{V}},
        ytick={1,10,100,1000,10000, 100000, 1000000},
        ymode = log,
        ymax = 6000000,
        xmin = 0.5,
        xmax = 9.5,
        xtick style={draw=none},
        yminorticks=false,
        yticklabel style={xshift=4pt}
     ]
        \addplot[
            black,
            boxplot prepared = {
                lower whisker = 1.00,
                lower quartile = 1.62,
                median = 19.00,
                average = 146.67,
                upper quartile = 133.66,
                upper whisker = 817.00,
                box extend = 0.8,
            }
        ] coordinates {};

        \addplot[
            black,
            boxplot prepared = {
                lower whisker = 1.00,
                lower quartile = 11.67045125,
                median = 108.34,
                average = 70666.71,
                upper quartile = 6902.31,
                upper whisker = 1022766.00,
                box extend = 0.8,
            }
        ] coordinates {};

        \addplot[
            black,
            boxplot prepared = {
                lower whisker = 1.00,
                lower quartile = 19.50,
                median = 131.14,
                average = 100594.50,
                upper quartile = 7118.50,
                upper whisker = 1022766.00,
                box extend = 0.8,
            }
        ] coordinates {};

        \addplot[
            black,
            boxplot prepared = {
                lower whisker = 1.0,
                lower quartile = 2.325892857,
                median = 19.00,
                average = 2268.68,
                upper quartile = 75.87,
                upper whisker = 23244.68,
                box extend = 0.8,
            }
        ] coordinates {};

        \addplot[
            black,
            boxplot prepared = {
                lower whisker = 1.0,
                lower quartile = 11.19473872,
                median = 40.00,
                average = 458.22,
                upper quartile = 452.88,
                upper whisker = 3703.75,
                box extend = 0.8,
            }
        ] coordinates {};

        \addplot[
            black,
            boxplot prepared = {
                lower whisker = 1,
                lower quartile = 42.25,
                median = 817.00,
                average = 174781.53,
                upper quartile = 111782.00,
                upper whisker = 1142204.00,
                box extend = 0.8,
            }
        ] coordinates {};

        \addplot[
            black,
            boxplot prepared = {
                lower whisker = 1,
                lower quartile = 12.88889391,
                median = 20.00,
                average = 727.98,
                upper quartile = 157.27,
                upper whisker = 5473.33,
                box extend = 0.8,
            }
        ] coordinates {};

        \addplot[
            black,
            boxplot prepared = {
                lower whisker = 1.00,
                lower quartile = 1.239637257,
                median = 15.21,
                average = 1160.00,
                upper quartile = 68.170,
                upper whisker = 10104.00,
                box extend = 0.8,
            }
        ] coordinates {};

        \addplot[
            black,
            boxplot prepared = {
                lower whisker = 1.00,
                lower quartile = 1.51339195,
                median = 8.18,
                average = 357.33,
                upper quartile = 56.41,
                upper whisker = 4133.00,
                box extend = 0.8,
            }
        ] coordinates {};
    \end{axis}
 \end{tikzpicture} \end{minipage} \\
    \end{tabular}

    
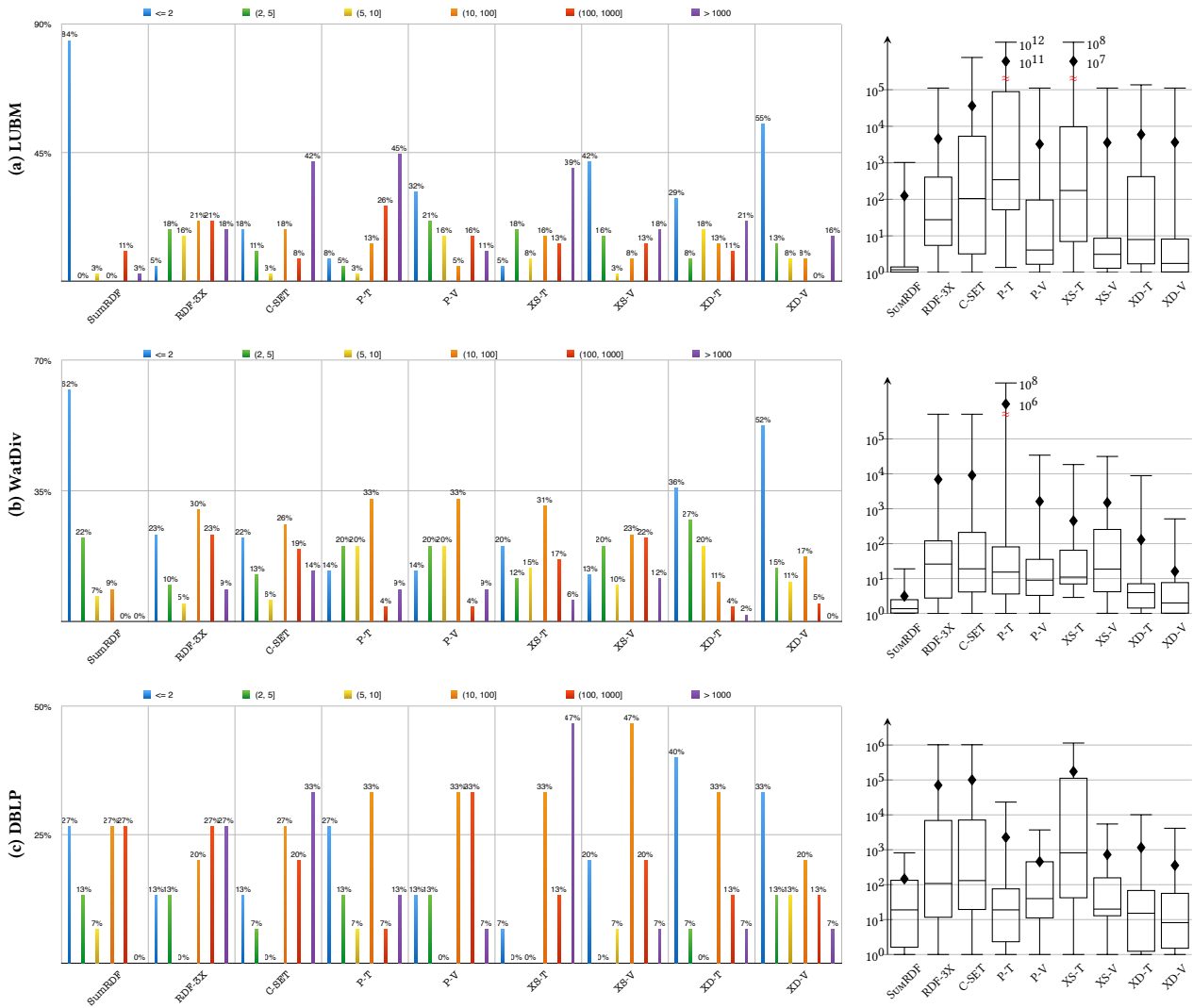
\captionof{figure}{Histograms and box plots of the distribution of the q-error on test datasets}\label{fig:qerror}
\end{table*}

We implemented our techniques in a prototype system called
\our.\footnote{\url{http://www.cs.ox.ac.uk/isg/tools/SumRDF/}} The system is
written in Java. It evaluates queries over graphs stored in RAM using left-deep
index nested loop joins. We evaluated our system with three distinct
objectives. First, we compared the accuracy of the estimates produced by \our
and several state-of-the-art RDF and relational systems. Second, we
investigated the impact of the correct handling of $\mu$-unifiable queries;
these correspond to self-joins and are difficult to estimate. Third, we
evaluated the usefulness of the bounds provided by
Theorem~\ref{th:qerrorbound}. We ran our experiments on a Dell computer with
96~GB of RAM, Fedora 22, and two Xeon X5667 processors with 16 physical cores.

\subsection{Test Data}

We used three well-known RDF benchmarks shown in Table~\ref{table:data}. We
group queries by their shape as linear, star, snowflake (joined stars), and
complex (e.g., cyclic). For each group, we show the number of all queries, the
number of $\mu$-unifiable queries, and the minimal and maximal numbers of atoms
in a query. Each query is assigned an ID. All datasets and queries are
available on the \our Web site.

\myparagraph{LUBM}~\cite{DBLP:journals/ws/GuoPH05} is a synthetic benchmark
comprising a data generator, an OWL~2~DL ontology, and 14 test queries. The
ontology must be taken into account for the queries to produce results, so we
generated an RDF graph with 500 universities and extended it with triples
implied by the ontology, thus obtaining 91.1~M triples. The 14 queries are
relatively simple, and are all $\mu$-unification-free, so we additionally
handcrafted 24 queries, five of which are $\mu$-unifiable.

\myparagraph{WatDiv}~\cite{DBLP:conf/semweb/AlucHOD14} is another synthetic
benchmark designed to stress-test RDF systems. We used its data generator to
obtain an RDF graph with 108.9~M triples. The benchmark also includes three
queries and 17 query templates with one parameter, and a generator that
replaces the parameter with resources from the RDF graph. We instantiated each
query template into five concrete queries and, if the parameter was not in the
class position, we obtained one more query by replacing the parameter with a
fresh variable. We thus obtained 103 queries overall, all of which are
$\mu$-unification-free.

\myparagraph{DBLP} is a bibliography database that was converted into an RDF
graph of 55.5~M triples. We used six queries already been considered in the
literature \cite{Zou:2014:GGS}, and we handcrafted nine additional queries. We
thus obtained 15 queries, eight of which are $\mu$-unifiable.

\subsection{Constructing the Summaries}

LUBM can be easily partitioned into ten partitions, and the numbers of
connections do not vary much across resources so we used just one bucket per
histogram. The typed summary for LUBM contained fewer than 15,000 edges, so the
refinement step was not needed. In contrast, WatDiv and DBLP could not be
effectively partitioned, so we used just one partition type, and we manually
selected between one and 64 histogram buckets per predicate. The typed
summaries were large so we further refined them using ${\mdim = 20}$, ${\ndim =
2}$, and $\mathit{target}$ of 45,000 for WatDiv and 30,000 for DBLP.
Table~\ref{table:data} shows that each graph was compressed by a factor more
than 1,000.

\subsection{Comparison Systems}

We evaluated our system against a portfolio of the following four systems,
which represent the state of the art in open-source and commercial (RDF) data
management.

\myparagraph{\rdfx} \cite{Neumann2010} v0.3.7 is a state of the art RDF system
whose cardinality estimator was specifically tailored to RDF and was shown to
outperform the related approaches.

\myparagraph{\cset} is the \emph{characteristic sets} technique
\cite{Neumann:2011} that specifically targets RDF and pays particular attention
to star queries. There is no publicly available implementation of this
technique, so we reimplemented it ourselves. It was shown to outperform the
approach by \citet{Stocker:2008}, so we did not test the latter separately.

\myparagraph{PostgreSQL} v9.4.5 uses a conventional cardinality estimator with
one-dimensional histograms. Statistics collection was automatic.

\myparagraph{SystemX} is a commercial RDBMS.\footnote{Publishing benchmarks
with the system named is prohibited by the system's license.} Its highly tuned
cardinality estimator operates in two modes: in the \emph{static mode} the
system uses statistics collected automatically prior to estimation, and in the
\emph{dynamic mode} it can sample the data during estimation if the estimate
accuracy is low. We used the maximum sampling level.

\smallskip

To store RDF data into RDBMSs, we replaced all resources with unique integers,
and then we considered two storage modes. First, we stored all triples in one
ternary \emph{triple table}, so all queries amount to self-joins over this
table. Second, we \emph{vertically partitioned} the data by storing $\tuple{s,
\rdftype, o}$ as tuple $\tuple{s}$ in a unary relation $o$, and $\tuple{s, p,
o}$ with ${p \neq \rdftype}$ as tuple $\tuple{s,o}$ in a binary relation $p$.
Using multiple tables generally allows for more accurate statistics collection.
We thus obtained a total of eight comparison systems: \rdfx, \cset, two
(\postgres{T} and \postgres{V}) PostgreSQL variants, and two static
(\system{S}{T} and \system{S}{V}) and two dynamic (\system{D}{T} and
\system{D}{V}) SystemX variants. For systems other than \cset, we obtained
their cardinality estimates using their \texttt{EXPLAIN} facility.

\subsection{Accuracy Experiments}\label{sec:evaluation:accuracy}

Figure~\ref{fig:qerror} shows the q-error distribution per dataset and system,
thus summarising the results of our accuracy experiments on all 156 test
queries. Each bar chart shows the percentages of the queries with q-error in
the intervals from the legend. Each box plot groups the q-errors per system:
the box shows the lower and upper quartiles of the q-error, the line dividing
the box shows the median, the line whiskers show the minimum and maximum, and
the diamond mark shows the average. Please note that the maximum q-errors
sometimes fall beyond the shown portion of the $y$ axis.

As one can see from the figure, \our outperforms the other systems by several
orders of magnitude on LUBM and WatDiv. On DBLP, it exhibits the best maximum
and average q-errors, and SystemX with dynamic sampling is the only system that
provides comparable performance. We conjecture the latter to be due to
selections in queries: if a query contains a resource that has not been
captured accurately using histograms, more accurate statistics are obtained on
demand. Overall, our technique considerably improves the state of the art in
cardinality estimation of RDF queries.

We noticed no pattern in the systems' behaviour on queries with q-error above
10: each system undershoots on some, but overshoots on other queries. 
Different variants of the same system often differ greatly on the same query;
for example, \postgres{T} might overshoot while \postgres{V} would undershoot,
and similarly for SystemX. Also, for each dataset and system, there is a
query where the system is most accurate, but \our is never the least accurate.

Atoms of the form $\tuple{x,\rdftype, C}$, whose presence in the query does not
affect the query answers, were a common source of estimation errors. In \cset,
\postgres{T}, and \system{S}{T}, such atoms impose a selection on the
$\rdftype$ property that is combined using independence assumption, which
usually leads to a significant underestimation.

We further analysed \emph{difficult queries}, where \our exhibited a q-error
above 10. Five queries of LUBM have q-errors between 750 and 1,000; three are
triangular and two are cyclic. Nine queries of WatDiv have q-error between 10
to 20, and all of these contain one or two selections. Finally, on DBLP, three
queries have q-errors between 10 and 20, and additional five have q-errors
above 20; all are either cyclic or have up to three selections. Thus, queries
containing cycles and/or selections are sometimes hard for \our. However, most
systems exhibited a q-error above 10 on these queries as well. SystemX with
dynamic sampling performed better on acyclic queries with selections: the
system would additionally sample whenever the statistics were insufficiently
precise.

\subsection{Impact of $\mu$-Unification}

\begin{table}[tb]
    \centering
    \footnotesize
    \setlength{\tabcolsep}{1.1pt}
    \def\arraystretch{1.3}
    \newcommand{\pK}{\phantom{k}}
    \newcommand{\vr}[1]{\rotatebox[origin=c]{90}{\;\;#1\;\;}}
    \newcommand{\hr}[1]{\multicolumn{1}{c|}{\rotatebox[origin=c]{90}{\;\;#1\;\;}}}
    \caption{The q-error for $\mu$-unifiable queries}\label{tab:qerror:unifiable}
    \begin{tabular}{|r|r|*{10}{r|}}
    \hline
                                & \hr{ID}   & \hr{Eq.\ \eqref{eq:freeexpect}}   & \hr{\our} & \hr{\rdfx}    & \hr{\cset}    & \hr{\postgres{T}} & \hr{\postgres{V}} & \hr{\system{S}{T}}    & \hr{\system{S}{V}}    & \hr{\system{D}{T}}    & \hr{\system{D}{V}} \\
    \hline
    \multirow{5}{*}{\vr{LUBM}}  & U1        & 1.4~k                             & 769.8     & 71.3~\pK      & 205.2~k       & 3.5~T             & 122.5~\pK         & 26.3~G                & 2.6~k                 & 17.5~k                & 3.4~k              \\
                                & U2        & 111.0~k                           & 926.3     & 111.0~k       & 111.0~k       & 7.9~T             & 111.0~k           & 12.3~M                & 111.0~k               & 982.0~\pK             & 111.0~k            \\
                                & U3        & 1.1~\pK                           & 1.4       & 254.5~\pK     &  67.5~k       & 25~G              & 2.7~\pK           & 87.3~M                & 1.2~\pK               & 4.8~k                 & 1.0~\pK            \\
                                & U4        & 1.2~\pK                           & 1.8       & 455.4~\pK     & 763.9~k       & 650~G             & 9.8~\pK           & 49.2~M                & 1.1~\pK               & 2.7~k                 & 1.2~\pK            \\
                                & U5        & 1.1~\pK                           & 1.0       & 5.9~\pK       & 1.3~\pK       & 11.2~k            & 1.6~\pK           & 10.2~k                & 1.5~\pK               & 12.2~k                & 1.5~\pK            \\
    \hline
    \multirow{8}{*}{\vr{DBLP}}  & D1        & 1.2~\pK                           & 1.2       & 9.9~k         & 440.8~k       & 106.3~\pK         & 3.4~\pK           & 220.4~k               & 1.4~\pK               & 2.8~\pK               & 2.8~\pK            \\
                                & D2        & 2.1~\pK                           & 2.1       & 1.0~M         & 1.0~M         & 23.2~k            & 40.0~\pK          & 1.0~M                 & 8.8~\pK               & 15.2~\pK              & 3.2~\pK            \\
                                & D3        & 238.0~\pK                         & 19.9      & 11.9~k        & 238.0~\pK     & 45.4~\pK          & 238.0~\pK         & 238.0~\pK             & 238.0~\pK             & 79.3~\pK              & 238.0~\pK          \\
                                & D4        & 320.2~\pK                         & 320.2     & 3.9~k         & 29.7~k        & 432.1~\pK         & 1.0~k             & 213.5~k               & 5.5~k                 & 6.3~k                 & 55.8~\pK           \\
                                & D5        & 721.3~\pK                         & 721.3     & 313.3~\pK     & 4.1~k         & 3.9~\pK           & 344.4~\pK         & 4.1~k                 & 4.1~k                 & 2.0~\pK               & 4.1~k              \\
                                & D6        & 57.0~\pK                          & 57.0      & 57.0~\pK      & 57.0~\pK      & 2.7~\pK           & 57.0~\pK          & 57.0~\pK              & 57.0~\pK              & 57.0~\pK              & 57.0~\pK           \\
                                & D7        & 210.3~\pK                         & 210.3     & 10.1~k        & 10.1~k        & 10.1~k            & 561.3~\pK         & 10.1~k                & 76.5~\pK              & 10.1~k                & 8.2~\pK            \\
                                & D8        & 817.0~\pK                         & 817.0     & 817.0~\pK     & 817.0~\pK     & 38.9~\pK          & 817.0~\pK         & 817.0~\pK             & 817.0~\pK             & 817.0~\pK             & 817.0~\pK          \\
    \hline
    \multicolumn{2}{|c|}{min}   & 1.1~\pK                                       & 1.0       & 5.9~\pK       & 1.3~\pK       & 2.7~\pK           & 1.6~\pK           & 57.0~\pK              & 1.1~\pK               & 2.0~\pK               & 1.0~\pK            \\
    \multicolumn{2}{|c|}{med}   & 210.3~\pK                                     & 57.0      & 817.0~\pK     & 29.7~k        & 10.1~k            & 122.5~\pK         & 213.5~k               & 76.6~\pK              & 982.1~\pK             & 55.8~\pK           \\
    \multicolumn{2}{|c|}{avg}   & 8.8~k                                         & 296.1     & 90.1~k        & 204.3~k       & 940.8~G           & 8.8~k             & 31.8~M                & 9.5~k                 & 4.3~k                 & 9.2~k              \\
    \multicolumn{2}{|c|}{max}   & 111.0~k                                       & 926.2     & 1.0~M         & 1.0~M         & 7.9~T             & 111.0~k           & 263.4~M               & 111.0~k               & 17.6~k                & 111.0~k            \\
    \hline
    \end{tabular}
\end{table}

Formula \eqref{eq:freeexpect} for $\mu$-unification-free queries does not
require computing the partition base for the query, so the formula should be
easier to evaluate than the general formula from Theorem~\ref{th:expect}.
Hence, we investigated whether using the general formula improves estimation
accuracy. Table~\ref{tab:qerror:unifiable} compares the q-error of the estimate
computed using formula \eqref{eq:freeexpect} with the q-errors of all other
systems (including \our) for all $\mu$-unifiable queries. WatDiv does not occur
in the table since all of its queries are $\mu$-unification-free.

As one can see, $\mu$-unifiable queries tend to be difficult: for each query,
at least one system exhibits a q-error above 900. However, handling
$\mu$-unifiable queries correctly reduces the maximum error by two orders of
magnitude, and it also considerably improves the median and the average.
Finally, \our considerably outperforms all other systems on all aggregate
metrics, apart from the median of \system{D}{V}, whereas the aggregate metrics
of formula \eqref{eq:freeexpect} seem similar to \postgres{V}. Thus, our
principled approach to $\mu$-unifiable queries contributes substantially to
estimation accuracy.

\subsection{Probabilistic Error Bounds}\label{sec:evaluation:bounds}

For $\varepsilon$ equal to 10, 100, and 1,000, Table~\ref{table:std} shows the
bound on the probability of q-error exceeding $\varepsilon$ computed as stated
in Theorem~\ref{th:qerrorbound}. The probabilities (expressed as percentages)
are grouped into five intervals. Each table cell shows the number of queries in
the interval, and in parentheses the number of queries where the q-error indeed
exceeds $\varepsilon$. For example, on LUBM, there are 20 queries for which the
probability of q-error being larger than 10 is less than 1\%, and for two of
these the q-error is larger than 10.

We could not compute $\std{q}$ in several cases: query ${q \cup \rho(q)}$ from
Theorem~\ref{th:variance} is always $\mu$-unifiable and tends to be very
complex. Moreover, Theorem~\ref{th:qerrorbound} does not apply to queries
with ${\expected{q} < 1}$. The numbers in parentheses next to each dataset name
show the number of queries on which we could compute $\prob{\qerror \geq
\varepsilon}$.

Already for ${\varepsilon = 10}$, we can correctly bound with 99\% certainty
the estimates for 18 out of 29 LUBM queries, 15 out of 48 WatDiv queries, and
two out of three DBLP queries. With ${\varepsilon = 100}$, we cover 75\% of
LUBM, 87\% of WatDiv, and all of the DBLP queries. These results suggest that
the confidence bounds provided by Theorem~\ref{th:qerrorbound} are indeed very
useful. In our future work, we shall incorporate them into query planning
algorithms so that a planner can prefer plans with higher confidence
cardinality estimates.

\section{Conclusion and Outlook}\label{sec:conclusion}

We presented a new approach for estimating the cardinality of CQs on RDF
graphs. Our approach is based on graph summarisation and it formalises the
estimation problem based on a precise `possible worlds' semantics. Our
technique can also provide statistical confidence for an estimate, even for
arbitrarily shaped, complex CQs. Finally, we showed experimentally that our
approach outperforms state of the art RDF and relational cardinality estimators.

We see many exciting opportunities for future work. On the theoretical side,
supporting range queries should be easy, but adding \texttt{DISTINCT} and
aggregation is likely to be more involved. We shall also try to incorporate
further information about the data and thus reduce the number of graphs
represented by a summary. On the practical side, we shall further analyse the
sources of errors and develop ways of computing more precise summaries,
possibly by incorporating variants of dynamic sampling. We shall also
investigate how to update the summary without recomputing it from scratch when
the underlying graph is updated.

\begin{acks}
The research presented in this paper was supported by the
\grantsponsor{EPSRC}{EPSRC}{http://www.epsrc.ac.uk/} projects MaSI$^3$
(\grantnum{EPSRC}{EP/P025943/1}), DBOnto (\grantnum{EPSRC}{EP/L012138/1}), and
ED$^3$ (\grantnum{EPSRC}{EP/N014359/1}).
\end{acks}

\bibliographystyle{ACM-Reference-Format}
\balance
\bibliography{references}

\ifappendix{
    \onecolumn
    \clearpage
    \appendix
    \section{Proof for Section 3}

In this appendix, we prove Proposition \ref{prop:consistency-checking} and Theorem \ref{th:qerrorbound}.

\subsection{Proof of Proposition \ref{prop:consistency-checking}}

To prove Proposition \ref{prop:consistency-checking}, we start by defining the
inverse of the summarisation function $\mu$ in a summary $\summary =
\summarytuple$.

\begin{definition}\label{def:inverse-mapping}\em
    For $\summary = \summarytuple$ a summary, the \emph{inverse function}
    $\inv{\mu}$ of $\mu$ maps each bucket ${b \in \res{\sgraph}}$ and each triple
    $\striple \in \sgraph$ as follows:
    \begin{align*}
        \inv{\mu}(b)        & = \setof{\gres \in \dom{\mu} \mid \mu(\gres) = b},\\
        \inv{\mu}(\striple) & = \setof{\gtriple \in \dom{\mu} \times \dom{\mu} \times \dom{\mu} \mid \mu(\gtriple) = \striple}.
    \end{align*}
\end{definition}

Please note that, by Definition \ref{def:summary} of size $\size{\cdot}$, we have
that $\size{b} = |\inv{\mu}(b)|$ and $\size{\striple} = |\inv{\mu}(\striple)|$, for
each $b \in \res{\sgraph}$ and each $\striple \in \sgraph$. The next lemma shows
that the consistency of a summary can be determined as described in Proposition
\ref{prop:consistency-checking}; in addition, it shows how to count the number
of graphs represented by a consistent summary.

\begin{lemma}\label{lemma:consistent-count}
For $\summary = \summarytuple$ a summary, $\summary$ is consistent if and only
if $\weight{\striple} \leq \size{\striple}$ for each ${\striple \in \sgraph}$.
Furthermore, when $\summary$ is consistent, the following holds:
\begin{equation}
    |\semantics| = \prod_{\striple \in \sgraph} \binom{\size{\striple}}{\weight{\striple}}. \label{eq:consistent-count}
\end{equation}
\end{lemma}
\begin{proof}
Let $\summary = \summarytuple$ be a summary. Then, by Definition
\ref{def:semantics} of a graph represented by $\summary$ and by Definition
\ref{def:inverse-mapping} of $\inv{\mu}$, for each graph $\graph$, we have that
${\graph \in \semantics}$ if and only if
\begin{enumerate}[label = (\alph{*})]
    \item \label{item:repr1} $\graph \subseteq \bigcup_{\striple \in \sgraph} \inv{\mu}(\striple)$, and
    \item \label{item:repr2} ${\weight{\striple} = |\inv{\mu}(\striple) \cap \graph|}$ for each ${\striple \in \sgraph}$.
\end{enumerate}
We next show that $\summary$ is consistent if and only if ${\weight{\striple}
\leq \size{\striple}}$ for each ${\striple \in \sgraph}$.

\smallskip
    
$(\Rightarrow)\;$ Assume that $\summary$ is consistent, and so ${\semantics
\neq \emptyset}$. Consider an arbitrary RDF graph ${\graph \in \semantics}$ and
an arbitrary triple ${\striple \in \sgraph}$. Then, $\weight{\striple}$ is the number
of triples in $\inv{\mu}(\striple)$ that occur in $\graph$. Since ${\size{\striple} =
|\inv{\mu}(\striple)|}$, we have that ${\weight{\striple} \leq \size{\striple}}$.

\smallskip

$(\Leftarrow)$\; Assume that $\weight{\striple} \leq \size{\striple}$ for each
$\striple \in \sgraph$. Then, for each ${\striple \in \sgraph}$, let $\graph_{\striple}$
be an arbitrary, but fixed subset of ${\inv{\mu}(\striple)}$ of size
$\weight{\striple}$; since $\weight{\striple} \leq \size{\striple}$, such
$\graph_{\striple}$ exists. Finally, let $G = \bigcup_{\striple \in \sgraph}
\graph_{\striple}$. By properties \ref{item:repr1} and \ref{item:repr2}, we have that
$G \in \semantics$.

\medskip

Therefore, $\summary$ is consistent if and only if ${\weight{\striple} \leq
\size{\striple}}$ for each ${\striple\in \sgraph}$. To conclude the proof, we
assume that $\summary$ is consistent and show that Equation
\eqref{eq:consistent-count} holds. By properties \ref{item:repr1} and
\ref{item:repr2}, each graph ${\graph \in \semantics}$ can be generated by
independently choosing, for each triple ${\striple \in \sgraph}$, a set
$\graph_{\striple}$ of atoms of size $\weight{\striple}$ such that ${\graph_{\striple}
\subseteq \inv{\mu}(\striple)}$. For each ${\striple \in \sgraph}$, there are
precisely
\begin{displaymath}
    \binom{\size{\striple}}{\weight{\striple}}
\end{displaymath}
distinct such sets $\graph_{\striple}$. So, we can count the total
number of RDF graphs in $\semantics$ as shown in Equation
\eqref{eq:consistent-count}.
\end{proof}

\subsection{Proof of Theorem \ref{th:qerrorbound}}

\qerrorbound*
\begin{proof}

First, note that since $\varepsilon > 1$,
$$
\frac{\varepsilon \cdot \std{q}}{(\varepsilon - 1) \cdot \expected{q}} > \frac{\std{q}}{(\varepsilon - 1) \cdot \expected{q}},
$$ 
so the additional inequality implies the main one. Therefore, it is enough to prove the first bound for the case when $\expected{q} > \varepsilon > 1$ and the second bound for the case when $\varepsilon \geq \expected{q} \geq 1$. We do this in the rest of the proof.

\medskip

Let $\expected{q} > \varepsilon > 1$. The inequality
 \begin{displaymath}
\qerror =    \max \left( \frac{ \max(|\ans{q}{\graph}|, 1)}{\max(\expected{q}, 1)}, \frac{\max(\expected{q}, 1)}{ \max(|\ans{q}{\graph}|, 1)}\right) < \varepsilon
 \end{displaymath}
is equivalent to 
$$
\frac{\max(\expected{q}, 1)}{\varepsilon} < \max(|\ans{q}{\graph}|, 1) < \varepsilon \cdot \max(\expected{q}, 1)
$$
and, since $\expected{q} > \varepsilon > 1$, to 
$$
\frac{\expected{q}}{\varepsilon} < |\ans{q}{\graph}| < \varepsilon \cdot \expected{q}.
$$
Therefore, since $\varepsilon \cdot \expected{q} - \expected{q} > \expected{q} - \myfrac{\expected{q}}{\varepsilon}$, 
\begin{multline*}
\prob{\qerror < \varepsilon} = 
\prob{\frac{\expected{q}}{\varepsilon} < |\ans{q}{\graph}| < \varepsilon \cdot \expected{q} } = 
\prob{\frac{\expected{q}}{\varepsilon} - \expected{q} < |\ans{q}{\graph}| - \expected{q} < \varepsilon \cdot \expected{q} - \expected{q}} > {} \\
\prob{\frac{\expected{q}}{\varepsilon} - \expected{q} < |\ans{q}{\graph}| - \expected{q} < \expected{q} - \frac{\expected{q}}{\varepsilon}} = 
\prob{\left| |\ans{q}{\graph}| - \expected{q} \right| < \expected{q} - \frac{\expected{q}}{\varepsilon}} = 
\prob{\left| |\ans{q}{\graph}| - \expected{q} \right| < \frac{(\varepsilon - 1) \cdot \expected{q}}{\varepsilon}} > \\
1 - \left(\frac{\varepsilon \cdot \std{q}}{(\varepsilon - 1) \cdot \expected{q}}\right)^2.
\end{multline*}
The last step follows from Chebyshev's inequality ${\prob{\big| |\ans{q}{\graph}| - \expected{q}\big| \geq k \cdot \std{q}} \leq
\frac{1}{k^2}}$ for $k = \frac{(\varepsilon - 1) \cdot \expected{q}}{\varepsilon \cdot \std{q}}$. This immediately implies the required bound
     $$
     {\prob{\qerror \geq \varepsilon} \leq \left(\frac{\varepsilon \cdot \std{q}}{(\varepsilon - 1) \cdot \expected{q}}\right)^2}.
    $$

\medskip

Let now $\varepsilon \geq \expected{q} \geq 1$. The inequality
$\qerror < \varepsilon$
is again equivalent to 
$$
\frac{\max(\expected{q}, 1)}{\varepsilon} < \max(|\ans{q}{\graph}|, 1) < \varepsilon \cdot \max(\expected{q}, 1)
$$
and, since $\varepsilon \geq \expected{q} \geq 1$, to $
|\ans{q}{\graph}| < \varepsilon \cdot \expected{q}$.
Therefore, 
\begin{multline*}
\prob{\qerror < \varepsilon} = 
\prob{ |\ans{q}{\graph}| < \varepsilon \cdot \expected{q} } = 
\prob{ |\ans{q}{\graph}| - \expected{q} < \varepsilon \cdot \expected{q} - \expected{q}} \geq {} \\
\prob{\expected{q} - \varepsilon \cdot \expected{q} < |\ans{q}{\graph}| - \expected{q} < \varepsilon \cdot \expected{q} - \expected{q}} = 
\prob{\left| |\ans{q}{\graph}| - \expected{q} \right| < \varepsilon \cdot \expected{q} - \expected{q}} = 
\prob{\left| |\ans{q}{\graph}| - \expected{q} \right| < (\varepsilon - 1) \cdot \expected{q}} > \\
1 - \left(\frac{\std{q}}{(\varepsilon - 1) \cdot \expected{q}}\right)^2.
\end{multline*}
The last step follows from Chebyshev's inequality for $k = \frac{(\varepsilon - 1) \cdot \expected{q}}{\std{q}}$. This immediately implies the required bound
     $$
     {\prob{\qerror \geq \varepsilon} \leq \left(\frac{\std{q}}{(\varepsilon - 1) \cdot \expected{q}}\right)^2}.
    $$
This concludes the proof of the theorem.
 \end{proof}

\section{Proofs for Section 4}

In this appendix, we prove Theorems \ref{th:expect} and \ref{th:variance}, and
Proposition \ref{prop:free-expect}. For the rest of this section, we fix an
arbitrary consistent summary ${\summary = \summarytuple}$ and a CQ $q$ such
that ${\res{q} \subseteq \dom{\mu}}$.

\subsection{Properties of the Partition Base}

We start by formalising the properties of the unifiable partitions of $q$ that
were intuitively outlined in Section \ref{sec:estimation:intuition}. To this
end, we first define the notion of an expansion w.r.t.\ a partition $P$ of an
answer $\tau$ to $\mu(q)$ over $\sgraph$.

\begin{definition}\label{def:expansion}\em
    Let $\tau$ be an answer to $\mu(q)$ on $\sgraph$ and let $P \in
    \pbase_\tau$ be a partition. A substitution $\pi$ with $\dom{\pi} =
    \var{q}$ and $\rng{\pi} \subseteq \dom{\mu}$ is a \emph{$\tau$-expansion}
    to $P$ if, for each ${x \in \var{q}}$,
    \begin{enumerate}[label=(\arabic*)]
        \item $\pi(x) \in \inv{\mu}(\tau(x))$,
        \item \label{item:subst1} $\psubst_{P}(x) \in \var{q}$ implies that $\pi(x) = \pi(\psubst_{P}(x))$,
    
        \item \label{item:subst2} $\psubst_{P}(x) \in \res{q}$ implies that $\pi(x) = \psubst_{P}(x)$.
    \end{enumerate}
    Then, $\expa[\tau]{P}$ contains each $\tau$-expansion to $P$, and
    $\mexpa[\tau]{P}$ contains each ${\pi \in \expa[\tau]{P}}$ for which no
    ${P' \in \pbase_\tau}$ exists such that ${P \prec P'}$ and ${\pi \in
    \expa[\tau]{P'}}$.
\end{definition}

\begin{proposition}\label{prop:base}
    For each answer $\tau$ to $\mu(q)$ on $\sgraph$ and each partition ${P \in
    \pbase_\tau}$, the following properties hold:
    \begin{enumerate}[label=(P\arabic*)]
        \item \label{item:base1} for each ${\pi \in \expa[\tau]{P}}$, exactly one partition ${P'
        \in \pbase_\tau}$ exists such that ${P \preceq P'}$ and ${\pi \in
        \mexpa[\tau]{P'}}$,
        
        \item \label{item:base2} for each ${\pi \in \mexpa[\tau]{P}}$ and all
        atoms $\qatom,\qatom' \in q$, there exists $u \in P$ with
        $\qatom,\qatom'\in u$ if and only if $\pi(\qatom) = \pi(\qatom')$.
    \end{enumerate}
\end{proposition}
\begin{proof}
For the rest of this proof, we let $\tau$ be an arbitrary answer to $\mu(q)$ on
$\sgraph$.

\medskip

\emph{Property \ref{item:base1}}~~Since $\pbase_\tau$ is finite and $\preceq$
is a partial order on $\pbase_\tau$, for each ${P \in \pbase_\tau}$ and each
${\pi \in \expa[\tau]{P}}$, there exists a partition ${P' \in \pbase_\tau}$ such
that ${P \preceq P'}$ and ${\pi \in \mexpa[\tau]{P'}}$. Hence, we next show
that ${\mexpa[\tau]{P_1} \cap \mexpa[\tau]{P_2} = \emptyset}$ for all
distinct partitions ${P_1, P_2 \in \pbase_\tau}$.

Consider two arbitrary distinct partitions ${P_1,P_2 \in \pbase_\tau}$. By the
definition of $\mexpa[\tau]{P_1}$ and $\mexpa[\tau]{P_2}$, if ${P_1 \prec P_2}$
or ${P_2 \prec P_1}$ holds, we then have ${\mexpa[\tau]{P_1} \cap
\mexpa[\tau]{P_2} = \emptyset}$. Thus, in the rest of this proof, we consider
the case in which ${P_1 \not \prec P_2}$ and ${P_2 \not \prec P_1}$.

Let $\approx$ be the smallest equivalence relation on (the atoms of) $q$ such
that ${\qatom \approx \qatom'}$ for all atoms ${\qatom,\qatom' \in q}$ such
that ${\qatom,\qatom' \in u}$ for ${u \in P_i}$ with ${i \in \interval{1}{2}}$.
Then, let $P$ be the set of all the equivalence classes of $\approx$. Since
$P_1$ and $P_2$ are partitions of $q$ and $\approx$ is an equivalence relation
on the atoms of $q$, we have that $P$ is a partition of $q$ as well. By the
definition of $\approx$, for each ${i \in \interval{1}{2}}$ and each ${u \in
P_i}$, there exists an equivalence class of $\approx$ that contains each atom
in $u$; thus, we have that ${P_i \preceq P}$. By the initial assumptions, we
have that ${P_1 \neq P_2}$, ${P_1 \not \prec P_2}$, and ${P_2 \not \prec P_1}$;
therefore, we also have that ${P_1 \prec P}$ and ${P_2 \prec P}$. We next show
that ${P \in \pbase_\tau}$ and, for each ${\pi \in \expa[\tau]{P_1} \cap
\expa[\tau]{P_2}}$, we have that ${\pi \in \expa[\tau]{P}}$. Please note that
this suffices to show that ${\mexpa[\tau]{P_1} \cap \mexpa[\tau]{P_2} =
\emptyset}$ since ${P_1 \prec P}$ and ${P_2 \prec P}$.

Consider an arbitrary substitution $\pi$ with ${\dom{\pi} = \var{q}}$ and
${\rng{\pi} \subseteq \dom{\mu}}$ such that ${\pi \in \expa[\tau]{P_1} \cap
\expa[\tau]{P_2}}$. Since $\pi$ is a $\tau$-expansion to $P_1$ and $P_2$, for
each $x \in \var{q}$, we have that ${\tau(x) = \mu(\pi(x))}$. Next, let
$\mathcal{G}_P$ be the term graph for $P$ defined in
Definition~\ref{def:partition}. To prove that $\pi$ is a $\tau$-expansion to
$P$, we first show that, for all terms $s, t \in \term{q}$ such that $s$
reaches $t$ in $\mathcal{G}_{P}$, we have that ${\pi(s) = \pi(t)}$. To this
end, we show that the following property holds for all atoms $\qatom, \qatom'
\in q$:

\begin{enumerate}[label=(A\arabic{*})]
     \item \label{item:aux-sat-1} $\qatom \approx \qatom'$ implies that
     $\pi(\qatom) = \pi(\qatom')$.
\end{enumerate}

To prove auxiliary property \ref{item:aux-sat-1}, we proceed by induction on
the number of steps required to derive ${\qatom \approx \qatom'}$. For the base
case, the empty relation $\approx_0$ clearly satisfies the property. For the
inductive step, consider an arbitrary relation $\approx_n$ obtained in $n$
steps that satisfies the property; we show that the same holds for all the
equivalences derivable from $\approx_n$. Since the equality relation $=$ is
reflexive, symmetric, and transitive, the derivation of $\qatom \approx_{n+1}
\qatom'$ due to reflexivity, symmetry, or transitivity clearly preserves the
required property. So, we next consider an arbitrary $i \in \interval{1}{2}$, a
set $u \in P_i$, and two atoms ${\qatom, \qatom' \in u}$, so that we derive
${\qatom \approx_{n+1} \qatom'}$. Let $\qatom = \tuple{s, p, o}$ and let
$\qatom' = \tuple{s', p', o'}$. Then, by the definition of $\mathcal{G}_{P_i}$,
terms $s$, $p$, and $o$ reach $s'$, $p'$, and $o'$ in $\mathcal{G}_{P_i}$,
respectively. Thus, by the definition of $\psubst_{P_i}$, we also have that
${\psubst_{P_i}(s) = \psubst_{P_i}(s')}$, ${\psubst_{P_i}(p) =
\psubst_{P_i}(p')}$, and ${\psubst_{P_i}(o) = \psubst_{P_i}(o')}$. By
properties \ref{item:subst1} and \ref{item:subst2} in
Definition~\ref{def:expansion}, we have that ${\pi(\qatom) = \pi(\qatom')}$, as
required.

Please note that $\mathcal{G}_P$ is obtained by adding an undirected edge between
$s_1$ and $s_2$, $p_1$ and $p_2$, and $o_1$ and $o_2$ for each $u \in P$ and
all atoms $\tuple{s_1, p_1, o_1}, \tuple{s_2, p_2, o_2} \in u$. By the
definition of $P$, we have that $\tuple{s_1, p_1, o_1} \approx \tuple{s_2, p_2,
o_2}$ for each such pair of atoms. So, because the equality relation $=$ is
reflexive, symmetric, and transitive, property \ref{item:aux-sat-1} implies
that $\pi(s) = \pi(t)$ for all terms that are reachable in $\mathcal{G}_P$.

\smallskip

Since $\pi(s) = \pi(t)$ for all $s, t \in \term{q}$ such that $s$ reaches $t$
in $\mathcal{G}_P$, partition $P$ is unifiable, and so $P \in B$. We next show
that $P$ is satisfied by $\tau$ and that $\pi$ satisfies conditions
\ref{item:subst1} and \ref{item:subst2} in Definition~\ref{def:expansion}. For
$x \in \var{q}$ a variable, we next consider two cases.
\begin{itemize}
    \item Assume that $\psubst_{P}(x)$ is a variable $y \in \var{q}$. By the
    definition of $\psubst_P$, we have that $x$ reaches $y$ in $\mathcal{G}_P$,
    and so $\pi(x) = \pi(y)$. Since $\tau(x) = \mu(\pi(x))$ and $\tau(y) =
    \mu(\pi(y))$, we conclude that $\tau(x) = \tau(y)$.
    
    \item Assume that $\psubst_{P}(x)$ is a resource $\gres \in \res{q}$. By
    the definition of $\psubst_P$, we have that $x$ reaches $\gres$ in
    $\mathcal{G}_P$, and so $\pi(x) = \gres$. Since $\tau(x) = \mu(\pi(x))$, we
    conclude that $\tau(x) = \mu(\gres)$.
\end{itemize}
Therefore, $P$ is satisfied by $\tau$ and $\pi$ satisfies conditions \ref{item:subst1} and \ref{item:subst2}
in Definition~\ref{def:expansion}; thus, $\pi$ is a $\tau$-expansion to $P$, as required.

\medskip

\emph{Property \ref{item:base2}}~~ Let $P \in \pbase_\tau$ be a partition and
let ${\pi \in \mexpa[\tau]{P}}$ be a $\tau$-expansion; we show that, for all
$\qatom, \qatom'\in q$, there exists $u \in P$ with $\qatom, \qatom'\in u$ if
and only if $\pi(\qatom) = \pi(\qatom')$.

\smallskip

$(\Rightarrow)$~By properties \ref{item:subst1} and \ref{item:subst2} in
Definition \ref{def:expansion} and by Definition \ref{def:partition} of
$\psubst_{P}$, we have that $\pi(\qatom) = \pi(\qatom')$ for each $u \in P$ and
all ${\qatom,\qatom'\in u}$, as required.

\smallskip

$(\Leftarrow)$ For the sake of a contradiction, assume that two atoms $\qatom,
\qatom' \in q$ exist such that $\pi(\qatom) = \pi(\qatom')$, but no ${u \in P}$
exists such that $\qatom, \qatom' \in u$. Then, let $P'$ be the partition of
$q$ such that, for all ${\qatom_1, \qatom_2 \in q}$, there exists $u'\in P'$
such that ${\qatom_1, \qatom_2 \in u'}$ if and only if ${\pi(\qatom_1) =
\pi(\qatom_2)}$. Please note that, by the definition of $P'$, a set $u' \in P'$
exists such that ${\qatom, \qatom' \in u'}$.

We first show that $P \prec P'$. To this end, consider an arbitrary $u \in P$
and two arbitrary atoms $\qatom_1, \qatom_2 \in u$. Because $\pi$ is a
$\tau$-expansion to $P$, we have that ${\pi(\qatom) = \pi(\qatom')}$, and so a
set $u' \in P'$ exists such that $\qatom_1, \qatom_2 \in u'$. Thus, we have
that ${P \preceq P'}$. Furthermore, by the initial assumption, there exists no
${u \in P}$ such that ${\qatom, \qatom' \in u}$; hence, we also have that ${P
\prec P'}$.

By the construction of $P'$, we have that ${\pi(s) = \pi(t)}$ for all terms $s,
t \in \term{q}$ such that $s$ reaches $t$ in $\mathcal{G}_{P'}$. So, $P'$ is a
unifiable partition of $q$ and satisfies properties \ref{item:subst1} and
\ref{item:subst2} of Definition \ref{def:expansion}. Since $\pi$ is a
$\tau$-expansion to $P$, we have that $\tau(x) = \mu(\pi(x))$ for each $x \in
\var{q}$; thus, $P'$ is satisfied by $\tau$ and $\pi$ is a $\tau$-expansion to
$P'$. This is a contradiction, since $P \prec P'$ and $\pi \in
\mexpa[\tau]{P}$. \end{proof}

\smallskip

\subsection{Counting Expansions}

We next show how one can effectively count, given an answer $\tau$ to $\mu(q)$
over $\sgraph$ and a partition $P \in \pbase_\tau$, the number of
$\tau$-expansions to $P$. To this end, we first an auxiliary lemma.

\begin{lemma}\label{lemma:kappa}
    For all partitions $P, P' \in \pbase$ with $P \prec P'$,
    \begin{equation}
        \label{eq:kappa} - \dcoeff{P}{P'} ~~ = \sum_{P \prec P'' \preceq P'} \dcoeff{P''}{P'}.
    \end{equation}
\end{lemma}
\begin{proof}
Let $P$ and $P'$ be as specified in the lemma. By Definition \ref{def:partition},
we have that
\begin{displaymath}
    - \dcoeff{P}{P'}    \ \ =\ \ -(\Ceven{P}{P'} - \Codd{P}{P'}) \ \ = \ \ \Codd{P}{P'} - \Ceven{P}{P'}.
\end{displaymath}
Because $P$ and $P'$ are distinct partitions, by Definition \ref{def:partition} of
a chain from $P$ to $P'$, the number of even chains from $P$ to $P'$ in
$\pbase$ can be computed by summing, for each $P'' \in B$ with $P \prec P''$
and $P'' \preceq P'$, the number of odd chains from $P''$ to $P'$ in $\pbase$.
Similarly, the number of odd chains from $P$ to $P'$ in $B$ can be computed by
summing, for each $P'' \in B$ with $P \prec P''$ and $P'' \preceq P'$, the
number of even chains from $P''$ to $P'$ in $B$. Hence, we have that
\begin{align*}
   \Ceven{P}{P'}  &~~= \sum_{P \prec P'' \preceq P'} \Codd{P''}{P'}, \text{ and}\\
   \Codd{P}{P'}  &~~= \sum_{P \prec P'' \preceq P'} \Ceven{P''}{P'}.
\end{align*}
Therefore, we have that
\begin{align*}
    - \dcoeff{P}{P'}   &= \Codd{P}{P'} - \Ceven{P}{P'} = \left(\sum_{P \prec P''  \preceq P'} \Ceven{P''}{P'}\right) - \left(\sum_{P \prec P'' \preceq P'} \Codd{P''}{P'}\right) \\[1em]
                       &= \sum_{\substack{P \prec P'' \preceq P'}} \left(\Ceven{P''}{P'} - \Codd{P''}{P'}\right) = \sum_{\substack{P \prec P'' \preceq P'}} \dcoeff{P''}{P'}.
\end{align*}
\end{proof}

\begin{proposition}\label{prop:count-expansions}
The following equations hold for each answer $\tau$
to $\mu(q)$ on $\res{\sgraph}$ and each $P \in \pbase_\tau$:
\begin{align}
    \label{eq:sol1} |\expa[\tau]{P}|       &~~~=~~~ \coeff{\tau}{P},\\
    \label{eq:sol2} |\mexpa[\tau]{P}|    &~~~= \sum_{P'\in \pbase_\tau, P \preceq P'} \dcoeff{P}{P'} \cdot \coeff{\tau}{P'}.
\end{align}
\end{proposition}
\begin{proof}
Let $\tau$ be as specified in the lemma, and let $\preceq_\tau$ be the
restriction of $\preceq$ to $\pbase_\tau$.

\medskip

We first prove Equation \eqref{eq:sol1}. To this end, consider a partition $P
\in \pbase_\tau$. By Definition \ref{def:expansion} of a $\tau$-expansion to
$P$, each ${\pi \in \expa[\tau]{P}}$ can vary only on the variables occurring
in the range of $\psubst_{P}$. Furthermore, for each $x \in \vrng{\psubst_P}$,
we have that $\pi(x) \in \inv{\mu}(\tau(x))$ and ${|\inv{\mu}(\tau(x))| =
\size{\tau(x)}}$; thus, there are ${\coeff{\tau}{P} = \prod_{x \in
\vrng{\psubst_{P}}} \size{\tau(x)}}$ distinct $\tau$-expansions to $P$.

\medskip 

We next prove Equation \eqref{eq:sol2} by induction on $\preceq_\tau$.

\smallskip

\textit{Base case.} Consider an arbitrary partition ${P \in \pbase_\tau}$ that
is $\preceq_\tau$-maximal. We then have that ${\mexpa[\tau]{P} =
\expa[\tau]{P}}$; furthermore, by Equation \eqref{eq:sol1}, we have that
${|\expa[\tau]{P}| = \coeff{\tau}{P}}$. In addition, we have that ${\sum_{P
\preceq_\tau P'} \dcoeff{P}{P'} \cdot \coeff{\tau}{P'} = \dcoeff{P}{P} \cdot
\coeff{\tau}{P}}$. Finally, by Definition \ref{def:partition} of $K$, we have that
${\dcoeff{P}{P} = 1}$; so, Equation \eqref{eq:sol2} holds.

\smallskip

\textit{Inductive Step.} Consider an arbitrary ${P \in \pbase_\tau}$
and assume that, for each ${P'' \in \pbase_\tau}$ with ${P \prec_\tau
P''}$, we have that
\begin{displaymath}
    |\mexpa[\tau]{P''}| ~= \sum_{P'' \preceq_\tau P'} \dcoeff{P''}{P'} \cdot \coeff{\tau}{P'};
\end{displaymath}
we show that the same holds for $P$. By the definition of $\mexpa[\tau]{P}$,
we have that ${\mexpa[\tau]{P} = \expa[\tau]{P} \setminus \bigcup_{P \prec_\tau
P''} \expa[\tau]{P''}}$. By Proposition \ref{prop:base}, for each ${\pi \in
\bigcup_{P \prec_\tau P''} \expa[\tau]{P''}}$, exactly one partition ${P'' \in
\pbase_\tau}$ exists such that ${P \prec_\tau P''}$ and ${\pi \in
\mexpa[\tau]{P''}}$. Therefore, we have that
\begin{align*}
    \mexpa[\tau]{P}    &~=~ \expa[\tau]{P} \setminus \bigcup_{P \prec_\tau P''} \mexpa[\tau]{P''}\text{, and}\\
    |\mexpa[\tau]{P}|  &~=~  |\expa[\tau]{P}| - \sum_{P \prec_\tau P''} |\mexpa[\tau]{P''}|.
\end{align*}
By Equation \eqref{eq:sol1}, we have that ${|\expa[\tau]{P}| =
\coeff{\tau}{P}}$. So, by the inductive hypothesis, we can compute ${|\mexpa[\tau]{P}|}$ as follows.
\begin{displaymath}
    |\mexpa[\tau]{P}| = \coeff{\tau}{P} -  \sum_{P \prec_\tau P''} \left( \sum_{P'' \preceq_\tau P'}  \dcoeff{P''}{P'} \cdot \coeff{\tau}{P'}\right) = \coeff{\tau}{P} -  \sum_{P \prec_\tau P'}  \coeff{\tau}{P'} \cdot  \left(\sum_{P \prec_\tau P'' \preceq_\tau P'}  \dcoeff{P''}{P'}\right).
\end{displaymath}
By Lemma \ref{lemma:kappa}, for each ${P' \in \pbase_\tau}$ with ${P \prec_\tau P'}$,
we have that $\sum_{P \prec_\tau P'' \preceq_\tau P'} \dcoeff{P''}{P'} = -\dcoeff{P}{P'}$; thus, we have that
\begin{align*}
    |\mexpa[\tau]{P}|     &= \coeff{\tau}{P} - \sum_{P \prec_\tau P'}  \coeff{\tau}{P'} \cdot \left( \sum_{P \prec_\tau P'' \preceq_\tau P'}  \dcoeff{P''}{P'}\right) \\[1em]
                                &= \coeff{\tau}{P} - \left(\sum_{P \prec_\tau P'} - \coeff{\tau}{P'} \cdot  \dcoeff{P}{P'}\right) =  \coeff{\tau}{P} + \left(\sum_{P \prec_\tau P'}  \coeff{\tau}{P'} \cdot  \dcoeff{P}{P'}\right).
\end{align*}
Finally, by Definition \ref{def:partition} of $K$, we have that $\dcoeff{P}{P}$ is $1$, so the following holds:
\begin{align*}
    |\mexpa[\tau]{P}|     &=  \coeff{\tau}{P} + \left(\sum_{P \prec_\tau P'}  \coeff{\tau}{P'} \cdot  \dcoeff{P}{P'}\right)\\[1em]
                                &=  \dcoeff{P}{P} \cdot \coeff{\tau}{P} + \left(\sum_{P \prec_\tau P'}  \dcoeff{P}{P'}\cdot \coeff{\tau}{P'} \right) = \sum_{P \preceq_\tau P'} \dcoeff{P}{P'}\cdot \coeff{\tau}{P'}.
\end{align*}
\end{proof}

\subsection{Proofs of Theorem \ref{th:expect} and Proposition
\ref{prop:free-expect}}

We are now ready to prove the correctness of the formulae for computing
$\expected{q}$ given in Theorem \ref{th:expect} and Proposition
\ref{prop:free-expect}.

\begin{lemma}\label{lemma:falling-factorial}
    For all $n, k, m \in \nat$ with $n \geq m \geq k > 0$,
    \begin{displaymath}
        \frac{\binom{n-k}{m-k}}{\binom{n}{m}} = \frac{(m)_k}{(n)_{k}}.
    \end{displaymath} 
\end{lemma}
\begin{proof}
Let $n, k, m$ be  natural numbers as specified in the lemma.
By the definition of binomial coefficient, we have that
\begin{align*}
    \binom{n}{m}     &= \frac{n!}{m!\cdot(n-m)!}, \text{ and}\\
    \binom{n-k}{m-k} &= \frac{(n-k)!}{(m-k)!\cdot(n-k - (m-k))!} = \frac{(n-k)!}{(m-k)!\cdot(n-m)!}.    
\end{align*}    
By using the standard laws of arithmetics, we then have that
\begin{displaymath}
    \frac{\binom{n-k}{m-k}}{\binom{n}{m}}   = \frac{\frac{(n-k)!}{(m-k)!\cdot(n-m)!}}{\frac{n!}{m!\cdot(n-m)!}} 
                                            = \frac{(n-k)! \cdot m!}{n! \cdot (m-k)!} 
                                            = \frac{m\cdots (m-k+1)}{n\cdots (n-k+1)} 
                                            = \frac{(m)_k}{(n)_k}.
\end{displaymath}
\end{proof}

\begin{lemma}\label{lemma:factor}
Let $\tau$ be an answer to $\mu(q)$ over $H$ and let $P \in
\pbase_\tau$ be a partition. The following holds for each
 ${\pi \in \mexpa[\tau]{P}}$:
\begin{equation}    
    \label{eq:factor} \factor{\tau}{P} = \expected{\pi(q)}.
\end{equation}
\end{lemma}
\begin{proof}
Let $\tau$ and $P$ be as specified in the lemma. Consider an arbitrary
substitution ${\pi \in \mexpa[\tau]{P}}$ and let ${q_\pi = \pi(q)}$; we show
that ${\factor{\tau}{P} = \expected{q_\pi}}$.  Note that, for each $\qatom \in q_\pi$,
we have that $\var{\qatom} = \emptyset$ and $\res{\qatom} \subseteq \dom{\mu}$. 

By the definition of $\tau$-expansion, we have that $\tau(x) = \mu(\pi(x))$ for
each $x \in \var{q}$, and so $\mu(q_\pi) = \tau(\mu(q))$. Because $P$ is
satisfied by $\tau$, for each ${u \in P}$ and all atoms ${\qatom, \qatom' \in
u}$, we have that ${\tau(\mu(\qatom)) = \tau(\mu(\qatom'))}$; so,
$\tau(\mu(u))$ is a set consisting of a single triple from $\sgraph$. In the
following, for each triple ${\striple \in \tau(\mu(q))}$, let
${\pcount{\striple} \in \nat}$ be as specified in Definition \ref{def:factor}.

By property \ref{item:base2} of Proposition \ref{prop:base}, for all atoms
${\qatom, \qatom' \in q}$, we have that ${\pi(\qatom) = \pi(\qatom')}$ if and
only if there exists $u \in P$ with ${\qatom, \qatom' \in u}$. That is, there
is a one-to-one correspondence between the triples in $q_\pi$ and the sets in
$P$. So, for each triple ${\striple \in \tau(\mu(q))}$, we have that
\begin{equation}
    \label{eq:atoms} \pcount{\striple}  ~~=~~ |\setof{u \in P \mid \tau(\mu(u)) = \setof{\striple}}| 
                                        ~~=~~ |\setof{\gtriple \in  q_\pi \mid \mu(\gtriple) = \striple}|.
\end{equation}
By Definition \ref{def:expect-var} of expectation, we have that
\begin{equation}
    \label{eq-expected-pi}\expected{q_\pi} ~= \sum_{\graph \in \semantics} \frac{|\ans{q_\pi}{\graph}|}{|\semantics|}.
\end{equation}
By the initial assumption, $\summary$ is consistent; so, ${|\semantics| > 0}$
and $\expected{q_\pi}$ is well-defined. Next, let ${\mathcal{N} = \sum_{\graph
\in \semantics} |\ans{q_\pi}{\graph}|}$. Since $q_\pi$ is a variable-free,
${|\ans{q_\pi}{\graph}|}$ is either $0$ or $1$; thus, ${\mathcal{N}}$ is the
number of graphs ${\graph \in \semantics}$ with ${q_\pi \subseteq \graph}$. By
Definition \ref{def:semantics} of a graph represented by $\summary$, for each
graph $\graph$, we have that ${\graph \in \semantics}$ 
 and $q_\pi \subseteq \graph$ if and only if
\begin{enumerate}[label = (\Alph{*})]
    \item \label{item:repr-1} $\graph \subseteq \bigcup_{\striple \in \sgraph} \inv{\mu}(\striple)$, 
    \item \label{item:repr-2} for each $\striple \in \sgraph$, the set ${\inv{\mu}(\striple) \cap \graph}$ contains each $\gtriple \in q_\pi$  with $\mu(\gtriple) = \striple$, and
    \item \label{item:repr-3} ${\weight{\striple} = |\inv{\mu}(\striple) \cap \graph|}$ for each ${\striple \in \sgraph}$.
\end{enumerate}
To prove the lemma, we next consider two alternative cases.

\medskip

\textit{(Case 1)}~ Assume that a triple $\striple \in \tau(\mu(q))$ exists such
that ${|\pcount{\striple}| > \weight{\striple}}$; so ${\factor{\tau}{P} = 0}$. By
Equation \eqref{eq:atoms}, $\pcount{\striple}$ is the number of triples in $q_\pi$
that $\mu$ maps onto $\striple$. So, there is no graph $\graph$ that satisfies
both properties \ref{item:repr-2} and \ref{item:repr-3} above, and thus
${\mathcal{N} = 0}$ and ${\expected{q_\pi}} = 0$. Therefore, ${\factor{\tau}{P}
= \expected{q_\pi}}$ and the lemma holds.

\medskip

\textit{(Case 2)}~ Assume that ${\pcount{\striple} \leq \weight{\striple}}$ for each ${\striple \in
\tau(\mu(q))}$; so 
\begin{displaymath}
    \factor{\tau}{P} = \prod_{\striple \in \tau(\mu(q))}\frac{(\weight{\striple})_{\pcount{\striple}}}{(\size{\striple})_{\pcount{\striple}}}.
\end{displaymath}
Each graph ${\graph \in \semantics}$ with ${q_\pi \subseteq \graph}$ can be
generated by independently choosing, for each ${\striple \in \sgraph}$, a subset
$G_\striple$ of $\inv{\mu}(\striple)$ such that $|G_\striple| = \weight{\striple}$ and
$G_\striple$ contains each $\gtriple \in q_\pi$ with $\mu(\gtriple) = \striple$.
Please recall that $\size{\striple}$ is the size of $\inv{\mu}(\striple)$. Next, we
 show how to compute, for each $\striple \in \sgraph$, the number
$\mathcal{N}_\striple$ of distinct such sets by considering two alternative cases.
\begin{itemize}
    \item Assume that ${h \in \tau(\mu(q))}$.  By Equation \eqref{eq:atoms},
    there are precisely $\pcount{\striple}$ distinct triples $\gtriple \in q_\pi$ such that ${\mu(\gtriple)} = \striple$, and so
    \begin{displaymath}
        \mathcal{N}_\striple = \binom{\size{\striple} - \pcount{\striple}}{\weight{\striple} - \pcount{\striple}}.
    \end{displaymath}

    \item Assume that ${h \not \in \tau(\mu(q))}$. Hence, no
    triple ${\gtriple \in q_\pi}$ exists such that ${\mu(\gtriple) = \striple}$. So, we have that
    \begin{displaymath}
        \mathcal{N}_\striple = \binom{\size{\striple}}{\weight{\striple}}.
    \end{displaymath}
\end{itemize}
Therefore, we can compute $\mathcal{N}$ as follows:
\begin{displaymath}
        \mathcal{N} \ \ = \sum_{\graph \in \semantics} |\ans{q_\pi}{\graph}| 
        \ = \ \prod_{\striple \in \sgraph} \mathcal{N}_\striple 
        \ \ = \prod_{\striple \in \tau(\mu(q))} \binom{\size{\striple} - \pcount{\striple}}{\weight{\striple} - \pcount{\striple}} \ \ \cdot  \prod_{\striple \in (\sgraph \setminus \tau(\mu(q)))}\binom{\size{\striple}}{\weight{\striple}}.
\end{displaymath}
By the definition of expected answer, we then have that
\begin{displaymath}
    \expected{q_\pi} \ \ = \sum_{\graph \in \semantics} \frac{|\ans{q_\pi}{\graph}|}{|\semantics|} 
                    \ \ = \ \ \frac{\mathcal{N}}{|\semantics|} 
                    \ \ = \ \ \frac{\prod_{\striple \in \tau(\mu(q))} \binom{\size{\striple} - \pcount{\striple}}{\weight{\striple} - \pcount{\striple}} \cdot  \prod_{\striple \in (\sgraph \setminus \tau(\mu(q)))}\binom{\size{\striple}}{\weight{\striple}}}{|\semantics|}.
\end{displaymath}
By using the formula for computing $|\semantics|$ from Lemma \ref{lemma:consistent-count},
we obtain the following:
\begin{align*}
    \expected{q_\pi}    &~=~  \frac{\prod_{\striple \in \tau(\mu(q))} \binom{\size{\striple} - \pcount{\striple}}{\weight{\striple} - \pcount{\striple}} \cdot  \prod_{\striple \in (\sgraph \setminus \tau(\mu(q)))}\binom{\size{\striple}}{\weight{\striple}}}{|\semantics|}
                        \ \ = \ \ \frac{\prod_{\striple \in \tau(\mu(q))} \binom{\size{\striple} - \pcount{\striple}}{\weight{\striple} - \pcount{\striple}} \cdot  \prod_{\striple \in (\sgraph \setminus \tau(\mu(q)))}\binom{\size{\striple}}{\weight{\striple}}}{\prod_{\striple \in \sgraph} \binom{\size{\striple}}{\weight{\striple}}}\\[1em]
                        &~=~ \frac{\prod_{\striple \in \tau(\mu(q))} \binom{\size{\striple} - \pcount{\striple}}{\weight{\striple} - \pcount{\striple}} \cdot  \prod_{\striple \in (\sgraph \setminus \tau(\mu(q)))}\binom{\size{\striple}}{\weight{\striple}}}{\prod_{\striple \in \tau(\mu(q))} \binom{\size{\striple}}{\weight{\striple}} \cdot \prod_{\striple \in (\sgraph \setminus \tau(\mu(q)))} \binom{\size{\striple}}{\weight{\striple}}}
                        \ \ = \prod_{\striple \in \tau(\mu(q))} \frac{\binom{\size{\striple} - \pcount{\striple}}{\weight{\striple} - \pcount{\striple}}}{\binom{\size{\striple}}{\weight{\striple}}}.
\end{align*}
Finally, by Lemma \ref{lemma:falling-factorial}, we have that
\begin{displaymath}
    \expected{q_\pi}  \ \  = \prod_{\striple \in \tau(\mu(q))} \frac{\binom{\size{\striple} - \pcount{\striple}}{\weight{\striple} - \pcount{\striple}}}{\binom{\size{\striple}}{\weight{\striple}}} \ \ = \prod_{\striple \in \tau(\mu(q))}\frac{(\weight{\striple})_{\pcount{\striple}}}{(\size{\striple})_{\pcount{\striple}}} \ \ = \ \ \factor{\tau}{P}.
\end{displaymath}
\end{proof}

We are now ready to prove Theorem \ref{th:expect}.

\expect*
\begin{proof}
In the rest of this proof, for each substitution $\pi$ with $\dom{\pi} =
\var{q}$ and $\rng{\pi} \subseteq \dom{\mu}$, we let $q_\pi = \pi(q)$;
furthermore, let $P_0 = \setof{\setof{\qatom} \mid \qatom \in q}$ be a
partition of $q$. By Definition \ref{def:partition} of partition base, we have
that $P_0 \in \pbase$; in addition, for each answer $\tau$ $\mu(q)$ on
$\sgraph$, we have that $P_0 \in \pbase_\tau$; finally, for each $P \in
\pbase$, we have that $P_0 \preceq P$.

Consider an arbitrary $\graph \in \semantics$; because $\summary$ is
consistent, such $\graph$ exists. Then, $|\ans{q}{\graph}|$ is the number of
substitutions $\pi$ with $\dom{\pi} = \var{q}$ and $\rng{\pi} \subseteq
\dom{\mu}$ such that ${q_\pi \subseteq \graph}$. Since ${\graph \in
\semantics}$, for each such $\pi$, the substitution $\tau$ such that ${\tau(x)
=\mu(\pi(x))}$, for each $x \in \dom{\pi}$, is an answer to $\mu(q)$ on
$\sgraph$. In addition, we have that $\pi \in \expa[\tau]{P_0}$. So, by
Definition~\ref{def:expect-var} of expected number of answers, the following
holds:
\begin{displaymath}
    \expected{q}  \ = \sum_{\graph \in \semantics} \frac{|\ans{q}{\graph}|}{|\semantics|} \ = \sum_{\graph \in \semantics} \sum_{\tau \in \ans{\mu(q)}{\sgraph}} \sum_{\pi \in \expa[\tau]{P_0}} \frac{|\ans{q_\pi}{\graph}|}{|\semantics|}.
\end{displaymath}
By Proposition \ref{prop:base}, we have that $\expa[\tau]{P_0} = \bigcup_{P \in
B_\tau, P_0 \preceq P} \mexpa[\tau]{P}$. Then, because ${P_0 \preceq P}$ for each $P \in \pbase_\tau$, we
have that
\begin{displaymath}
    \expected{q} \ = \sum_{\graph \in \semantics} \sum_{\tau \in \ans{\mu(q)}{\sgraph}} \sum_{\pi \in \expa[\tau]{P_0}} \frac{|\ans{q_\pi}{\graph}|}{|\semantics|} \ \ = \sum_{\graph \in \semantics} \sum_{\tau \in \ans{\mu(q)}{\sgraph}} \sum_{P \in \pbase_\tau} \sum_{\pi \in \mexpa[\tau]{P}} \frac{|\ans{q_\pi}{\graph}|}{|\semantics|}.
\end{displaymath}
So, by the definition of $\expected{q_\pi}$, we can compute
$\expected{q}$ as follows:
\begin{displaymath}
    \expected{q}  \  = \sum_{\tau \in \ans{\mu(q)}{\sgraph}} \sum_{P \in \pbase_\tau} \sum_{\pi \in \mexpa[\tau]{P}} \sum_{\graph \in \semantics}  \frac{|\ans{q_\pi}{\graph}|}{|\semantics|} \ \ = \sum_{\tau \in \ans{\mu(q)}{\sgraph}} \sum_{P \in \pbase_\tau} \sum_{\pi \in \mexpa[\tau]{P}} \expected{q_\pi}.
\end{displaymath}
By Lemma \ref{lemma:factor}, we then have that
\begin{displaymath}
    \expected{q}  = \sum_{\tau \in \ans{\mu(q)}{\sgraph}} \sum_{P \in \pbase_\tau}  \sum_{\pi \in \mexpa[\tau]{P}}  \factor{\tau}{P}.
\end{displaymath}
Finally, Proposition \ref{prop:count-expansions} shows how to count the number of $\preceq$-maximal $\tau$-expansions, so we can rewrite $\expected{q}$ as follows:
\begin{displaymath}
    \expected{q}  =  \sum_{\tau \in \ans{\mu(q)}{\sgraph}} \sum_{P \in \pbase_\tau}  \factor{\tau}{P} \ \ \cdot \sum_{P' \in \pbase_\tau, P \preceq P'} \dcoeff{P}{P'} \cdot \coeff{\tau}{P'}.
\end{displaymath}
\end{proof}

As pointed out in Section \ref{sec:estimation:formalisation}, Proposition
\ref{prop:free-expect} follows immediately from Theorem \ref{th:expect}.

\freeexpect*

\subsection{Proof of Theorem \ref{th:variance}}

We next show that our method for computing $\std{q}$ is correct.

\variance*
\begin{proof}
Let $\rho$ be a substitution mapping each variable in $\var{q}$ to a fresh
variable, and let ${q' = \rho(q)}$. We first show that equation
\eqref{eq:variance-rewr} holds for each RDF graph $\graph \in \semantics$.
\begin{equation}
    |\ans{q}{\graph}|^2 = |\ans{q \cup q'}{\graph}| \label{eq:variance-rewr}
\end{equation}
To this end, consider an arbitrary graph $\graph \in \semantics$. Then, $|\ans{q}{\graph}|^2$
is the cardinality of the set $\mathit{Pairs}$ that contains each pair
$\tuple{\pi, \pi'}$ of substitutions mapping the variables of $q$ to $\dom{\mu}$ such that
${\setof{\pi(q), \pi'(q)} \subseteq \graph}$. But then, Equation
\eqref{eq:variance-rewr} holds because the set of answers to ${q \cup q'}$ over
$\graph$ contains precisely, for each pair ${\tuple{\pi, \pi'} \in \mathit{Pairs}}$,
 one distinct substitution ${\kappa_{\pi, \pi'}}$ from ${\var{q}\cup
\var{q'}}$ to $\dom{\mu}$. Such $\kappa_{\pi,\pi'}$ can be obtained from $\pi$
and $\pi'$ by setting, for each ${x \in \var{q}}$,
\begin{displaymath}
    \kappa_{\pi,\pi'}(x)  = \pi(x) \text{ and } \kappa_{\pi,\pi'}(\rho(x)) = \pi'(x).
\end{displaymath}

\smallskip

We are now ready to prove the theorem. By Definition \ref{def:expect-var} of
variance, we have that
\begin{displaymath}
    \std{q}^2   = \sum_{\graph\in \semantics} \frac{|\ans{q}{\graph}|^2 - 2 \cdot |\ans{q}{\graph}| \cdot \expected{q}  + \expected{q})^2}{|\semantics|}.
\end{displaymath}
By using Equation \eqref{eq:variance-rewr} and by manipulating summations in
the usual way, we can express $\std{q}^2$ as follows.
\begin{displaymath}
    \sum_{\graph\in \semantics} \left[\frac{|\ans{q \cup q'}{\graph}|}{|\semantics|}\right] - 2 \cdot \expected{q}\cdot \sum_{\graph \in \semantics} \left[\frac{|\ans{q}{\graph}|}{|\semantics|}\right]  + \expected{q}^2
\end{displaymath}
By Definition \ref{def:expect-var} of expectation, we then have that
\begin{displaymath}
    \std{q}^2  \ \ = \ \ \expected{q\cup q'} - 2 \cdot \expected{q}^2    +  \expected{q}^2 \ \ = \ \ \expected{q\cup q'} -  (\expected{q})^2,
\end{displaymath}
and so the theorem holds.
\end{proof}

}{}

\end{document}